\documentclass[journal,twoside,web]{ieeecolor}
\usepackage{generic}
\usepackage{palatino,epsfig,fleqn,cite,color}
\usepackage{amsmath,amssymb,amsfonts}
\usepackage{bm,cases,subcaption} %
\usepackage{graphicx}
\usepackage{algorithm,algorithmic}
\usepackage{textcomp}

\newtheorem{prop}{Proposition}
\newtheorem{lemma}{Lemma}
\newtheorem{theorem}{Theorem}
\newtheorem{remark}{Remark}

\newtheorem{defn}{Definition}
\def\QED{~\rule[-1pt]{5pt}{5pt}\par}
\DeclareMathAlphabet{\matheur}{U}{eur}{m}{n}
\DeclareMathAlphabet{\matheurb}{U}{eur}{b}{n}
\DeclareMathAlphabet{\matheus}{U}{eus}{m}{n}
\DeclareMathAlphabet{\matheuf}{U}{euf}{m}{n}

\newcommand{\bfThe}{\mathbf{\Theta}}
\newcommand{\bfdel}{\mbox{\boldmath$\delta$}}
\newcommand{\bfDel}{{\bf\Delta}}

\newcommand{\hs}{\hspace{4mm}}

\newcommand{\IC}{\mathbb{C}}
\newcommand{\IR}{\mathbb{R}}
\newcommand{\IS}{\mathbb{S}}

\newcommand{\IU}{\mathbb{U}}

\newcommand{\II}{\mathbb{I}}

\newcommand{\diag}{{\rm diag}}

\newcommand{\phio}{\phi_p}
\newenvironment{mat}{\left[\begin{array}}{\end{array}\right]}

\newcommand{\red}{\color{black}}%
\newcommand{\blue}{\color{black}}%

\definecolor{bg}{rgb}{0,0.5,0.5}		%
\newcommand{\ti}{\color{black}}%
\definecolor{rg}{rgb}{0.4,0.6,0}		%

\newcommand{\rhoL}{{\ti\varrho_+}}		%

\DeclareMathAlphabet{\matheur}{U}{eur}{m}{n}
\DeclareMathAlphabet{\matheurb}{U}{eur}{b}{n}
\DeclareMathAlphabet{\matheus}{U}{eus}{m}{n}
\DeclareMathAlphabet{\matheuf}{U}{euf}{m}{n}

\newcommand{\LFT}{\mathcal{F}_\ell}

\newcommand{\Cp}{\mathcal{C}_p}
\newcommand{\Rp}{\mathcal{R}_p}
\newcommand{\Linf}{\mathcal{L}_{\infty}}
\newcommand{\Hinf}{\mathcal{H}_{\infty}}
\newcommand{\CLinf}{\mathcal{CL}_\infty} 	
\newcommand{\CHinf}{\mathcal{CH}_\infty}
\newcommand{\RLinf}{\mathcal{RL}_\infty} 	
\newcommand{\RHinf}{\mathcal{RH}_\infty}
\newcommand{\G}{\mathcal{G}}

\title
{Exact Robust Instability  Analysis 
for Networked Dynamical Systems with Biological Application}

\author{Shinji Hara$^{1}$~\IEEEmembership{Fellow, IEEE}, Yutaka Hori$^{2}$~\IEEEmembership{Senior Member, IEEE}, Tetsuya Iwasaki$^{3}$~\IEEEmembership{Fellow, IEEE}, Chung-Yao Kao$^{4}$~\IEEEmembership{Member, IEEE} and Sei Zhen Khong$^{4}$~\IEEEmembership{Senior Member, IEEE}
\thanks{This work was supported in part by the National Science and Technology Council of Taiwan (grant numbers: 
113-2221-E-110-048-MY3, 113-2918-I-110-003, 113-2222-E-110-002-MY3, 114-2622-8-110-00, 115-2218-E-007-003, and 115-2221-E-110-058-MY2).
The authors' names are listed in alphabetical order. Corresponding author: C.-Y. Kao.}
\thanks{$^{1}$ Shinji Hara is with 
Institute of Integrated Research, Institute of Science Tokyo, Tokyo 152-8550, Japan. {\tt shinjihara5202@gmail.com}}
\thanks{$^{2}$ ~Y.~Hori is with Applied Physics and Physico-Informatics, Keio University, 
3-14-1 Hiyoshi, Kohoku-ku, Yokohama, Kanagawa 223-8522, Japan.
{\tt yhori@keio.jp}. } 
\thanks{$^{2}$ T.~Iwasaki is with Mechanical and Aerospace Engineering, University of California Los Angels, 
420 Westwood Plaza, Los Angeles, CA 90095, USA. 
{\tt tiwasaki@ucla.edu}, }
\thanks{$^{4}$ Chung-Yao Kao and Sei Zhen Khong are with the Department of Electrical Engineering, National Sun Yat-sen University, 
          Kaohsiung 804201, Taiwan. {\tt \{cykao, szkhong\}@mail.nsysu.edu.tw} }
}

\begin{document}

\maketitle

\begin{abstract}
This paper investigates robust instability in nominally unstable uncertain networked dynamical systems, where all nominal agents share an identical single-input-single-output (SISO) linear time-invariant (LTI) system and each agent is subject to independent perturbations. This setting is motivated by the problem of sustaining periodic oscillations in nonlinear dynamics, for which exact analysis is generally intractable. We identify three classes of network structures—including cyclic and certain rank-deficient networks—for which the robust instability problem can be reduced to the analysis of a single representative SISO system. We derive sufficient conditions that exactly characterize the robust instability radius for these network classes. Finally, we demonstrate the practical utility of the proposed results by analyzing oscillatory behavior in a genetic regulatory network.
\end{abstract}


\section{Introduction} 
\label{sec:Intro}

Periodic oscillations arising from regulatory functions are fundamental and ubiquitous phenomena 
in biological systems. Prominent examples include synthetic repressilators~\cite{Elowitz2000}, 
spike-type periodic signaling in neuronal dynamics~\cite{FHNmodel,ijspeert:08}, and periodic 
pattern generation via Turing instability~\cite{YMKH:MBMC2015}. 
Some of these phenomena are related to instability {\ti of} %
stationary equilibria. 
Because exact mathematical models for biological systems are generally difficult to derive, 
approximate models are often utilized for analysis and synthesis. Consequently, much like robust 
stability analysis, robust instability analysis, which accounts for unmodeled dynamics and uncertainties, 
is essential to theoretically guarantee the preservation of desirable oscillatory behaviors.

Continuing authors' previous research on a simple feedback loop with scalar 
subsystems~\cite{hara2022instability, KHHIK:IFAC23, hara2023exact}, this paper investigates the robust 
instability of uncertain networked dynamical systems. {\ti The problem is to find a structured 
perturbation of the minimum norm among those that stabilize the nominally unstable network.} 
While the problem shares some conceptual ground with standard robust 
stability analysis, in particular, using $H_\infty$-norm as the measure of the size of uncertainties, they  
are fundamentally different. Mathematically, robust instability problem is equivalent to the strong 
stabilization problem~\cite{Youla:Automatica1974} with an additional task of minimizing the norm
of the controller. Despite the well-known Parity Interlacing Property (PIP)~\cite{Youla:Automatica1974}, which provides a 
necessary and sufficient condition for strong stabilization, finding a solution to the standard 
strong stabilization problem remains a challenge, let {\ti alone} a solution with minimum-norm. This 
clearly indicates that it is in general extremely difficult to find an exact solution for the robust 
instability problem. As such, our earlier work \cite{hara2023exact} {\ti focused} on 
exploring conditions under which the $H_\infty$ norm remains a valid measure for robust instability. 
This led to the development of the "Phase Change Rate (PCR) Maximization Problem," which offers 
a more tractable optimization framework for these complex systems.

The network version of robust instability analysis is the counterpart of the robust stability analysis discussed
in a recent paper~\cite{hara:19}. By the same reason mentioned before, it is far more difficult to derive
exact solutions to these problems compared to their stability analysis counterpart. In this paper, we primarily focus on certain cyclic 
networks, which are commonly seen in many biological applications, exploring conditions under which exact solutions can be found. 
The cyclic networks are special in that their circulant structure leads to homogeneous worst-case perturbation, 
which in turn allows for identifying one single-input-single-output (SISO) representative in regard to instability 
analysis. Namely, if this critical SISO agent is stabilized by a (SISO) perturbation, the same perturbation may be 
used to stabilize all other (unstable) agents in the network, leading consequently to a homogeneous stabilizing 
perturbation for the network. Such a feature simplifies instability analysis of a network to that of a SISO feedback 
loop with agent dynamics represented by transfer functions with complex coefficients, where the exact solutions are 
available from our previous work~\cite{hara2023exact}. The same feature is observed in some rank-deficient networks, 
where SISO representatives can be identified and utilized to derive exact solutions. The preliminary results 
necessary for establishing these exact solutions are developed in a companion paper~\cite{NetworkRIR_short}.

The main contributions of this paper are as follows. 
We identify three classes of network systems where robust instability analysis of such networks reduces to that of 
a representative SISO system. Further, we provide (sufficient) conditions for such systems to have 
{\ti the robust instability radius characterized exactly by a small gain condition}. 
One such class consists of cyclic networks where the inverse of the nominal agent dynamics is characterized by 
Hurwitz polynomials with specific constraints on the location of their roots. This particular class is found in biological
applications, serving as models for analyzing oscillatory behaviors of biomolecular reactions. Analysis of these models is performed
in this paper to illustrate our theoretical results. 

This paper is organized as follows. Notation and terminologies 
are summarized
{\ti right after this paragraph}.
Section~\ref{sec:PF} is devoted to the general problem formulation for robust 
instability analysis and some preliminary results. 
Section~\ref{sec:MarginalSS} provides marginal stability and 
stabilizability conditions for SISO 
{\ti complex} rational functions. 
The main focus of this paper, analysis of cyclic networks, 
is developed in Sections~\ref{sec:Cyclic1} and~\ref{sec:CycricExactRIR}. Section~\ref{sec:Cyclic1} describes the problem setup and 
important preliminary properties, while Section~\ref{sec:CycricExactRIR} develops conditions under which the robust  
instability radius of a cyclic network can be obtained exactly. 
Section~\ref{sec:LowRank} 
{\ti is devoted} to two special classes of rank-deficient networks. 
To demonstrate the practical relevance, we apply our cyclic case result  to analyze oscillatory biomolecular reactions 
of a genetic regulatory network (GRN) with a cyclic network structure in Section~\ref{sec:BioAppli}. 
Section~\ref{sec:Concl} summarizes the contributions of this paper and addresses some future research directions.
Proofs of certain technical results are provided in the appendix. 
\noindent 
\subsection*{Notation}
The set of integers is denoted by $\II$, of which the finite subset consisting of $1,\ldots,n$ is 
denoted by $\II_n$.
The set of real (complex) numbers is denoted by $\IR$ ($\IC$). Given a complex number $s$, 
$\Re(s)$ and $\Im(s)$ denote the real and imaginary parts of $s$, respectively. The open (closed) 
left and right half complex planes are abbreviated as OLHP (CLHP) and ORHP (CRHP), respectively.
We adopt the following notation regarding proper transfer functions:
\begin{itemize}
\item
$\Cp$ ($\Rp$): 
the set of proper rational transfer functions with complex (real) coefficients. 
\item
$\Linf$ ($\Hinf$): 
the set of transfer functions that are bounded on the imaginary axis $j\IR$ (analytic and bounded in the CRHP). 
\item 
$\mathcal{CL}_\infty$ ($\mathcal{RL}_\infty$): 
the set $\Cp$ ($\Rp$) $\cap \ \Linf$.
\item
$\mathcal{CH}_\infty$ ($\mathcal{RH}_\infty$): 
the set $\Cp$ ($\Rp$) $\cap \ \Hinf$.
\end{itemize}
{\red Superscripts are occasionally used on these symbols to clarify the dimensions of the intended quantities.}
The norms on $\Linf$ and $\Hinf$ are both denoted by
$\| \cdot \|_{\infty}$, since they coincide. In this paper, we will not distinguish a system from its transfer function
representation. It is well-known that a finite-dimensional {\ti linear time-invariant} system has a transfer function 
representation in $\Cp$. Such a system is referred to as ``stable'' if all its poles are in the OLHP, ``unstable'' if at least one of its poles is in the ORHP, 
and ``marginally stable'' if all its poles lie in the CLHP with {\ti at least one} pole on the imaginary axis {\ti and all imaginary-axis poles} are simple. In the case of a marginally stable system with a 
single pole at $j\omega_c$, we will refer to the system as $\omega_c$-marginally stable. It is also well known that stable systems from $\Cp$ form the 
subset $\CHinf$. 

For a complex function $f:\IC\rightarrow\IC$, the phase 
angle at point $s\in\IC$ where $f(s)\neq0$ is denoted by $\angle f(s)$, and 
$\theta_f(\omega):=\angle f(j\omega)$ denotes the phase of $f$ when restricted to the imaginary axis. Note that the angle $\angle f(s)$ of the complex number 
$f(s)$ is not uniquely determined but is chosen so that $\theta_f(\omega)$ is 
continuous on $\omega\in\IR$; such a choice is possible for $f\in{\RLinf}$ 
with no zeros on the imaginary axis.
The phase change rate (PCR) of $f(j\omega)$ with respect to $\omega$ is denoted by 
$\theta_f'$; i.e., $\theta_f'(\omega):=\frac{d}{d\omega}\theta_f(\omega)$.  

Let $\mathcal{F}_u(\cdot,\cdot)$ and $\mathcal{F}_\ell(\cdot,\cdot)$ denote the upper and lower linear 
fractional transformations {\red (LFT)}, respectively. {\ti The Kronecker product is denoted by $\otimes$.}

We use $[\![\star, \bullet]\!]$ to denote the transfer function
$
\begin{bmatrix} I & \star \\ \bullet & I 
\end{bmatrix}^{-1}.
$

\section{Problem Formulation}
\label{sec:PF}

\subsection{Uncertain Network Structure}
\label{subsec:UncertainDNS}

Consider the multi-agent system with $n$ uncertain dynamic SISO agents 
shown in Fig.~\ref{fig: LFT_UncertainNetwork} and described by 
\begin{equation} \label{sys}
\begin{mat}{c} z \\ y \end{mat}=H(s)
\begin{mat}{c} w \\ u \end{mat}, \hs
\begin{array}{l} w=\Delta(s) (z+d_\delta) + d_{h1}, \\ u=A(y+d_a) + d_{h2}, \end{array}
\end{equation}
\begin{equation} \label{sysH}
H(s) := W(s)\otimes I_n,
\end{equation}
where $z,w,y,u: \IR\mapsto\IR^n$ are internal signals in the system and 
$d_{h1}, d_{h2}, d_\delta, d_a: \IR\mapsto\IR^n$ are external signals inserted to define
internal stability of the whole system, denoted by $\Sigma(\Delta, H, A)$. 
Here $\Delta:=\diag(\delta_1,\cdots,\delta_n)$ with each $\delta_i$ denoting the dynamic 
uncertainty of the respective SISO agent, 
while $A\in\IR^{n\times n}$ is a constant matrix characterizing the interactions between agents. We will make the following standing assumption.

\begin{figure}
  \centering
\begin{picture}(172,129)(23,-25)
 \thicklines
 \put(75,25){\framebox(40,40){$H(s)$}}
 \put(160,4){$d_{h2}$}
 \put(20,4){$d_a$}
 \put(135,90){\circle{8}}
 \put(180,90){\vector(-1,0){40}}  
 \put(160,95){$d_{h1}$}
  \put(20,95){$d_\delta$}
 \put(56,90){\circle{8}}
 \put(12,90){\vector(1,0){40}}  
 \put(80,75){\framebox(30,30){$\Delta(s)$}}
  \put(35,-20){\color{blue}\framebox(120,90){}}
   \put(160,55){\color{blue} $G(s)$}
 \put(60,60){$z$}
 \put(55,55){\line(1,0){20}}  
 \put(55,55){\vector(0,1){31}}  
 \put(60,90){\vector(1,0){20}}  
 \put(123,60){$w$}
 \put(110,90){\vector(1,0){21}}  
 \put(135,55){\line(0,1){31}}  
 \put(135,55){\vector(-1,0){20}}  
 \put(80,-15){\framebox(30,30){$A$}}
 \put(60,38){$y$}
 \put(135,0){\circle{8}}
 \put(56,0){\circle{8}}
 \put(180,0){\vector(-1,0){40}} 
 \put(12,0){\vector(1,0){40}}  
 \put(135,90){\circle{8}}
 \put(55,35){\line(1,0){20}}  
 \put(55,35){\vector(0,-1){31}}  
 \put(60,0){\vector(1,0){20}}  
 \put(123,38){$u$}
 \put(110,0){\vector(1,0){21}}  
 \put(135,4){\line(0,1){31}}  
 \put(135,35){\vector(-1,0){20}}  
\end{picture}
\vspace{-2mm}
  \caption{Networked dynamical agents with uncertainties.} 
\label{fig: LFT_UncertainNetwork}
\vspace{-2mm}
\end{figure}

\vspace{0.1cm}
\noindent
{\bf [Assumption 1]}: $\Delta$ is stable (i.e., $\ti\Delta\in\RHinf^{n\times n}$).

With $W:=\begin{bmatrix} w_{11} & w_{12} \\  w_{21} & h \end{bmatrix}\in\RHinf^{2\times 2}$
The dynamics of each individual agent are characterized by the upper LFT
\begin{align*}
h_i:= \mathcal{F}_u(\delta_i, W) = h + \frac{w_{12}w_{21}\delta_i}{1 - w_{11}\delta_i},  
\end{align*}
where $h\in\RHinf$ represents the nominal dynamics shared commonly by all agents. Note that the LFT formulation 
allows us to model different types of ``perturbed'' systems; in this paper, we only focus on systems with 
the so-called \emph{multiplicative} perturbations, and make the following standing assumption.

\noindent
{\bf [Assumption 2]}: $W$ in~\eqref{sysH} has the form $\begin{bmatrix} 0 & w_m \\ h & h\end{bmatrix}$
with $h, \ w_m\in\RHinf$, $h$ is strictly proper and $w_m$ is minimum-phase. 
\vspace{0.1cm}

By Assumption 2, the agent dynamics $h_i$ have the form $h(1+w_m\delta_i)$. Such a setup allows us to consider 
heterogeneous agents while their nominal parts constitute a homogeneous network.

\begin{defn}[\cite{Vinnicombe2001}]\label{def:IS}
The uncertain network system $\Sigma(\Delta, H, A)$ 
is said to be internally stable if {\ti the} $4n \times 4n$ transfer function from the external input signal 
$[d_{h1}, d_{h2}, d_\delta, d_a]$ to the output signal $[z, y, w, u]$ is in $\RHinf^{4n \times 4n}$. 
\end{defn}

Under Assumptions~1 and~2, we may derive the following %
{\ti condition} for internal stability of $\Sigma(\Delta, H, A)$. 
\begin{prop} \label{prop:IS1}
Consider the system $\Sigma(\Delta, H, A)$ defined in~\eqref{sys} and~\eqref{sysH}. 
Suppose Assumptions 1 and 2 hold. Then $\Sigma(\Delta, H, A)$ is internally stable if and only if 
\begin{align}
\label{eq:IGIDelta}
\begin{bmatrix} I_n \\ G \end{bmatrix}
(I - \Delta G)^{-1}  \in \RHinf^{2n \times n}, \ 
\end{align} 
{\ti where $G:=\LFT(H,A)$}. %
\end{prop}
\begin{proof}
The readers are referred to~\cite{NetworkRIR_short} for a comprehensive analysis.
\end{proof} 

In this paper, we study the {\em robust instability} of the system $\Sigma(\Delta, H, A)$. 
Specifically, assuming that the nominal network $G:=\LFT(H,A)$ is unstable, we investigate whether 
$\Sigma(\Delta, H, A)$ remains unstable for all perturbations $\Delta$ within a given collection. To this end, 
we define 
\begin{align}
\begin{split}
&\bfDel_{d}:=\{\diag(\delta_1,\cdots,\delta_n)~|~ \delta_i\in\RHinf, {\ti~} i=1,\cdots,n\},\\
&\bfDel_{h}:=\{\delta I_n~|~ \delta\in\RHinf \}
\end{split}\label{def:Delta's}
\end{align}
and introduce the following perturbation classes:   
\begin{align*}
\IS_f(G) & :=  \{ \Delta \in \  \RHinf^{n\times n} ~|~ \mbox{\eqref{eq:IGIDelta} holds.} \}, \\
\IS_d(G) & :=  \IS_f(G) \cap \bfDel_{d}, \\ 
\IS_h(G) & :=  \IS_f(G) \cap \bfDel_{h} = \IS_d(G) \cap \bfDel_{h}. 
\end{align*}

Given a network system $\Sigma(\Delta,H,A)$, {\ti its} %
robust instability radius (RIR, for short), denoted by $\rho_*(\Sigma)$,
is defined as the smallest magnitude (in terms of $\RHinf$-norm) of the stable perturbation 
{\ti $\Delta\in\bfDel_d$} that internally 
stabilizes the system. Under Assumptions~1 and~2, Proposition~\ref{prop:IS1} implies that $\rho_*(\Sigma)=\rho_*(G)$, 
where $G:=\LFT(H,A)$ and
\begin{align} 
\label{rir}
\rho_*(G) := {\displaystyle\inf_{\Delta\in\mathbb{S}_d(G)}} ~\|\Delta\|_{\infty} {\ti.}
\end{align}
Furthermore, define 
\begin{align} 
\label{rir_hf}
\rho_h(G) := {\displaystyle\inf_{\Delta\in\mathbb{S}_h(G)}} ~\|\Delta\|_{\infty}, \quad
\rho_f(G) := {\displaystyle\inf_{\Delta\in\mathbb{S}_f(G)}} ~\|\Delta\|_{\infty}. 
\end{align}
Observe that the following inequalities hold: 
$${\red \varrho_p(G):=}\|G\|_{\infty}^{-1}\le \rho_f(G) \le \rho_*(G)\le\rho_h(G).$$
The first inequality holds by the small-gain theorem~\cite{Inoue:ECC2013}, while the other two follow the fact that 
$\IS_h(G)\subseteq\IS_d(G)\subseteq\IS_f(G)$. By convention, $\rho_*$, $\rho_h$, or $\rho_f$ is set to be infinity 
if the respective constraint set is empty. {\red Readers are referred to \cite[Theorem 1]{NetworkRIR_short} for more general 
relations between $\varrho_p$, $\rho_*$, $\rho_h$, and $\rho_f$.} 
It is well-known that $\IS_f(G)$ is empty if $G$ fails to satisfy the 
so-called ``parity interlacing property'' (PIP) condition~\cite{Youla:Automatica1974}; i.e., the number of real ORHP 
poles of $G$ between any pair of real zeros in the CRHP (including the zero at infinity) is even. Also note that if $G$ 
is a SISO system, then $\rho_f(G) = \rho_*(G)=\rho_h(G)$. 

The problem addressed in this paper is to characterize $\rho_*$ for given classes {\ti of} (unstable) network systems. 
In particular, we are interested in finding conditions under which %
{\red 
$\rho_*$ attains the value of one of the lower bounds described in Theorem 1 of \cite{NetworkRIR_short}. 
Note that in general we cannot expect that $\rho_* = \|G\|_{\infty}^{-1}$ .
}

To facilitate the subsequent development, we define the following SISO system classes: 
\begin{align*}
    &\G:= \{g\in\CLinf\ |\  g\text{ is strictly proper and unstable} \}, \\ 
    &\begin{aligned}
        \G_{n} :=\{g\in\G \ |\ & g\text{ has $n$ poles in the ORHP and satisfies}\\ 
        &\text{the PIP condition} \},
    \end{aligned}\\ 
    &\G_{n}^{\omega_p}:=\{g\in\G_{n}\ |\ |g(j\omega_p)|\!>\!|g(j\omega)|,\ \forall\omega\in\IR\backslash\{\omega_p\} \},
\end{align*}
and for $\omega_p\not=0$,
\begin{align*}
    &\G_{n}^{\pm\omega_p}:=\{g\in\G_{n}\cap\RLinf \ |\ |g(j\omega_p)|\!>\!|g(j\omega)|,\\
    &\hspace{5.5cm} \forall\omega\in\IR\backslash\{\pm\omega_p\} \}.
\end{align*}
\section{Marginal Stability and Stabilizability}
\label{sec:MarginalSS}

This section establishes the marginal stability and stabilizability conditions for SISO systems represented 
by complex rational functions, {\red which may arise from the complex eigenvalues of the interconnection matrix $A$. 
Marginal stabilization is instrumental in {\ti exactly} characterizing {\ti the} %
RIR of certain SISO systems~\cite{hara2023exact}.}
The results developed here will subsequently be utilized to obtain $\rho_*$ for specific networks. 
The conditions are based on a modified Nyquist criterion, generalizing the one in~\cite{hara2023exact}.

For $L\in\CLinf$, define
${\red \nu_L}(\epsilon):={\red \nu_{L+}}(\epsilon)-{\red \nu_{L-}}(\epsilon)$, where ${\red \nu_{L+}}(\epsilon)$ and 
${\red \nu_{L-}}(\epsilon)$ are the numbers of the counterclockwise and clockwise crossings of the 
semi-infinite interval $(1,+\infty)$ by the Nyquist plot of $L(j\omega+\epsilon)$; i.e., 
the contour $\{L(j\omega+\epsilon)~|~\omega\in(-\infty,\infty)\}$ for a given $\epsilon$.

\begin{lemma} \label{lemma:MS}
Consider a positive feedback system with loop transfer function $L\in\G_n$.
The system has all its poles in the CLHP
if and only if there exists $\epsilon_+>0$ such that the following two 
equivalent conditions hold for all $\epsilon\in(0,\epsilon_+)$. 
\begin{itemize}
\item[(i)] The number of counterclockwise encirclements of the Nyquist plot of 
$L(j\omega+\epsilon)$ about $1+j0$ is $n$.
\item[(ii)] ${\red \nu_L}(\epsilon)=n$ holds for the Nyquist plot of $L(j\omega+\epsilon)$.
\end{itemize}
Moreover, the system is marginally stable 
if and only if the following two conditions hold in addition to condition (i) or (ii): 
\begin{itemize}
\item[(iii)] $L(j\omega)=1$ for some $\omega\in\IR$.
\item[(iv)] $\left. \frac{d}{ds} L(s) \right|_{s=j\omega}\neq0$ for any $\omega$ 
such that $L(j\omega)=1$.
\end{itemize}
\end{lemma}
\begin{proof}
The result is almost identical to Lemma 4 in \cite{hara2023exact},
with the only difference being that we consider $L$ with complex coefficients 
here, which is immaterial to the proof given in~\cite{hara2023exact}.
\end{proof}

Based on Lemma~\ref{lemma:MS}, conditions for marginal stability of the feedback system with loop transfer
function $L$ are given as follows.
\begin{prop} \label{prop:OmegaMS2}
Consider a positive feedback system with loop transfer function $L\in\G_n$.
Given $\omega_p$, suppose $\frac{d}{d\omega}|L(j\omega)| =0$ at $\omega=\omega_p$.
Then, the feedback system is $\omega_p$-marginally stable if and only if 
condition (a) and one of conditions (b1) and (b2), indicated below, are satisfied. 
\begin{itemize}
\item[(a)] 
The loop transfer function $L$ satisfies the following:
\begin{equation*} %
\begin{array}{rcl}
L(j\omega_p) = 1, &\;& \left. \frac{d}{ds} L(s) \right|_{s=j\omega_p} \neq 0  \\
L(j\omega) \neq 1, &\;& \forall \; \omega \in\IR\backslash\{\omega_p\}.
\end{array}
\end{equation*} 
\end{itemize}
\begin{itemize}
\item[(b1)] 
The Nyquist plot of $L(j\omega)$ satisfies
${\red \nu_L}(0)=n-1$ and the phase change rate at $\omega = \omega_p$ is positive, i.e., $\theta_L'(\omega_p) > 0$.
\item[(b2)] 
The Nyquist plot of $L(j\omega)$ satisfies ${\red \nu_L}(0)=n$, 
and the phase change rate at $\omega = \omega_p$ is negative, i.e., $\theta_L'(\omega_p) <  0$. 
\end{itemize}
\end{prop}
\begin{proof} 
This result extends Proposition~2 in \cite{hara2023exact}, where $L$ is real-rational. Here, we allow $L$ to have complex coefficients. 
When $L$ is real-rational, the complex poles of the closed-loop system appear in complex conjugate pairs, resulting in symmetric critical frequencies 
at $\pm\omega_p$. Conversely, for systems with complex coefficients, such feature is generally lost; thus, we consider a single critical frequency 
$\omega_p$, which can be any real number. The proof of Proposition~\ref{prop:OmegaMS2} follows arguments analogous to those used for Proposition~2 
in \cite{hara2023exact}.
\end{proof}

Proposition~\ref{prop:OmegaMS2} yields the following conditions for the marginal stabilizability of a system $g$ by a controller $f$. Recall that 
internal stability requires all four transfer functions in the feedback connection $[\![g,f]\!]$ to be stable, which rules out any CRHP pole-zero 
cancellations between $g$ and $f$.

\begin{theorem}
\label{thm:MS}
We have the following necessary and sufficient conditions regarding marginal stabilizability of $g$ from $\G_n^{\omega_p}$
and $\G_n^{\pm\omega_p}$ (for $\G_n^{\pm\omega_p}$, $\omega_p\not=0$ by definition): 
\begin{itemize}
\item[(I)] Let $g\in\G_n^{\omega_p}$. Then $g$ can be marginally stabilized by $f \in \mathcal{RH}_\infty$ with
  $\|f\|_{\infty} = 1/\|g\|_{\infty}$ if and only if $n = 1$ and one of the following conditions holds:
  \begin{enumerate} \renewcommand{\theenumi}{\textup{(\roman{enumi})}}\renewcommand{\labelenumi}{\theenumi}
    \item $\omega_p = 0$, $g(0)$ is real, and $\theta_g'(0) > 0$;
    \item $\omega_p \ne 0$ and $\theta_g'(\omega_p) > |\sin(\theta_g(\omega_p))/\omega_p|$.
   \end{enumerate}
\item[(II)] Let $g\in\G_n^{\pm\omega_p}$. Then $g$ can be marginally stabilized by $f \in \mathcal{RH}_\infty$ with
  $\|f\|_{\infty} = 1/\|g\|_{\infty}$ if and only if $n = 2$ and $\theta_g'(\omega_p) > |\sin(\theta_g(\omega_p))/\omega_p|$. 
\end{itemize}
\end{theorem}

\begin{proof}
First, note that when $\omega_p \neq 0$, the assumption that $g$ has a unique peak gain on the entire imaginary axis at $j\omega_p$ (but not at $-j\omega_p$) 
implies $g \in \mathcal{L}_{\infty} \setminus \mathcal{R}\mathcal{L}_{\infty}$; that is, $g$ is not real-rational.

This claim extends Theorem 1 of \cite{hara2023exact} to the case where $g$ is complex-rational. In the real-rational setting of \cite{hara2023exact}, it holds 
that $|g(j\omega_p)| = \|g\|_{\infty}$ if and only if $|g(-j\omega_p)| = \|g\|_{\infty}$. This symmetry property is lost when 
$g$ is complex-rational, meaning that $\|g\|_{\infty} = |g(j\omega_p)| > |g(j\omega)|$ for all $\omega \in \mathbb{R} \setminus \{\omega_p\}$ typically holds . 
Consequently, the arguments in the proof of \cite[Theorem 1]{hara2023exact} can be applied to establish the claim by invoking Proposition~\ref{prop:OmegaMS2} 
instead of \cite[Proposition 2]{hara2023exact}. Specifically, the additional requirement that $g(0)$ be real when $\omega_p = 0$ arises because $f$ is real-rational; 
thus, $f(0)$ is real, and the condition $f(0)g(0) = 1$ can only hold if $g(0)$ is also real.
\end{proof}

\begin{remark}
Note that the minimum-norm compensator $f$ in question follows the structures established in~\cite[Theorem 1]{hara2023exact}. 
In the case where $\omega_p=0$, a solution is the constant function $f\equiv 1/\|g\|_{\infty}$, which clearly precludes the possibility of
pole-zero cancellations. In the case where $\omega_p\not=0$, a solution $f$ takes the form of (stable) first-order all-pass 
function. Should $g$ belong to {\ti $\G_2^{\pm\omega_p}$ or the subclass of $\G_1^{\omega_p}$ considered in this paper}, 
one can show there is no unstable
pole-zero cancellation between $g$ and $f$ by the following arguments.

Suppose that the unstable zero of the first-order all-pass $f$ cancels one unstable pole of $\ti g\in\G_2^{\pm\omega_p}$. 
This would leave the loop-gain function $L:=fg$ with one remaining unstable real pole. 
Since $\|L\|_{\infty} \le \|f\|_{\infty}\|g\|_{\infty} = 1$, the marginally stable poles of the closed-loop system $[\![L,1]\!]$ 
can only be at $\pm j\omega_p$. On the other hand, given the real unstable pole of $L$, the root-locus analysis dictates that there 
must be a real pole of the system $[\![k\cdot L,1]\!]$ crossing the origin when $k$ increases from $0$ to $1$; i.e. there exists 
$k\in(0,1)$ such that $k\cdot L(0) = 1$. This implies $L(0)=1/k > 1$, contradicting $\|L\|_{\infty} \le 1$. 

Finally, we note that in the case that $g\in\G_1^{\omega_p}$, the above arguments are no longer applicable. However, as shown in the 
subsequent sections, the plant $g$ under consideration in this paper has the form $g(s)=\lambda h(s)/(1-\lambda h(s))$, where 
$\lambda\in\IC$ and $h\in\Rp$. For $g\in\G_1^{\omega_p}$ in thisform, there can be no unstable pole-zero cancellation 
between $f$ and $g$, because the (only) unstable pole of $g$ cannot be of real-valued and hence cannot be cancelled by the real 
unstable zero of $f$. 
\end{remark}
\section{Cyclic Network: Preliminaries}
\label{sec:Cyclic1}
\subsection{System Setup}
\label{subsec:CycricProblem}

In this and the next sections we consider multi-agent system $\Sigma(\Delta,H,A)$ where 
agents are connected in a cyclic fashion. This corresponds to $A$ being a cyclic matrix of the 
form
$\begin{bmatrix}
    \ \ \mathbf{0}_{n-1}^T & -\mu_1 \\ \mathcal{U} & \ \mathbf{0}_{n-1}  \end{bmatrix}$,
where $\mathcal{U}:=\diag(-\mu_2,\cdots,-\mu_n)${\ti, $\mathbf{0}_{n-1}$ is the $n-1$ dimensional zero vector,} and we assume $\mu_1\mu_2\cdots\mu_n>0$.
Furthermore, we will set $w_m=1$ for simplicity ; as such, the dynamics of the agents 
are modelled by $(1+\delta_i(s))h(s)$, $i=1,\ldots,n$.

In this paper, we are interested in finding conditions for such networks to have the exact robust instability radius.
By Lemma~1 of~\cite{NetworkRIR_short}, internal stability of such networks is characterized by  
$\left(I - h(I+\Delta) A\right)^{-1} \in \RHinf^{n \times n}$,  
which is equivalent to the following condition:
\begin{align*}
\det \left(I - h(s)(I+\Delta(s))A) \right) \neq 0 , \; \forall \; \Re(s) \geq 0 . 
\end{align*}
Note that this is the stability condition for a multi-agent feedback system which has $n$ SISO subsystems
$\varPsi_i$, $i=1,\cdots,n$, interconnected cyclically as follows 
\begin{align*}
{\red \varPsi_1 = y_n\mapsto y_1, \mbox{ and }
\varPsi_i = y_{i-1}\mapsto y_i \mbox{ for  $i=2,\cdots, n$,} } 
\end{align*}
where each $\varPsi_i$ is LTI with the transfer function $\varPsi_i(s)=-\mu_i h(s)(1+\delta_i(s))$. 
Clearly, the characteristic equation of the system in question can be expressed as
\begin{equation} \label{checyc}
1 - (- \mu)^nh(s)^n\prod_{i=1}^n(1+\delta_i(s))=0, 
\end{equation}
where %
$\ti\mu:=(\mu_1\mu_2\cdots\mu_n)^{1/n}$ is a positive real number.
The stability of the system is guaranteed if and only if the above equation has no solution in 
the CRHP. 

One can readily verify that the characteristic equation of $A$ is $s^n - (-\mu)^n=0$. As such, 
the eigenvalues of $A$, denoted by $\lambda_k$, $k\in\mathbb{I}_n$, are given by
\begin{align}
\label{eq:eigs_A}
\lambda_k = (-\mu)\mathrm{e}^{j(2k\pi)/n}. 
\end{align}
Note that for any given $n$, $-\mu$ is always an eigenvalue of $A$, and the $n$ eigenvalues are equally spaced along a circle of radius $\mu$. 
Thus, the angle difference between two adjacent eigenvalues of $A$ is $2\pi/n$, and one
can readily verify that at least one eigenvalue of $A$ is in the closed first quadrant when $n\ge 2$; i.e., 
there are eigenvalues of $A$ with angles in $[0,\pi/2]$. Among these, let $\lambda_*$ 
denote the eigenvalue which has the smallest angle. It follows that 
\begin{align}
\label{def:lambda_star}
\lambda_* = 
\begin{cases}
\mu & \mbox{ if $n$ is even;} \\
\mu\mathrm{e}^{j\pi/n} & \mbox{ if $n$ is odd. }
\end{cases}
\end{align}
That is, $\angle\lambda_*=0$ when $n$ is even, while $\angle\lambda_*=\pi/n$ when $n$ is odd. 
For the subsequent development, we consider 
\begin{align}
g_k(s):=\frac{\lambda_kh(s)}{1-\lambda_kh(s)}
= \frac{\lambda_k}{\phi(s) - \lambda_k}
,
\label{def:gk}
\end{align}
where $\phi(s) := 1/h(s)$, and $\lambda_k$ are the eigenvalues of $A$ given in~\eqref{eq:eigs_A}. 
Let $\IU\subseteq\II_n$ denote the set of indices $k$ such that $g_k(s)$ is unstable. 
{\ti We will refer to any $\lambda_k$ with $k\in\IU$ as ``unstable'' eigenvalues.}

\subsection{Homogeneous Equivalence}
\label{subsec:EquiHeteroHomo}
This subsection establishes the identity $\rho_*=\rho_h$ by showing the equivalence of homogeneous and heterogeneous perturbations, 
enabling a known lower bound $\rhoL$ on $\rho_h$ \cite{NetworkRIR_short} to also bound $\rho_*$ from below.
For the subsequent development, $\log(\cdot)$ denotes the {\ti principal} logarithm function with domain $\mathbb{C}\backslash\{0\}$, defined as
\[
\log(z):=\ln |z|+j \theta
\]
where $\theta\in(-\pi,\pi]$ is the unique real number that satisfy $\cos\theta = \Re(z)/|z|$ and 
$\sin\theta = \Im(z)/|z|$. Note that the principal logarithm function is a bijection between its domain 
$\mathbb{C}\backslash\{0\}$ and the {\ti strip} $\mathbb{S}:=\{x+jy:x\in \mathbb{R}, \ y\in (-\pi,\pi]\}\subset\mathbb{C}$ 
of the complex plane. Specifically, for any $z\in\mathbb{C}\backslash\{0\}$, there is a unique $s\in\mathbb{S}$
such that $s=\log(z)$ and $z=\mathrm{exp}(s)$, where $\mathrm{exp}(\cdot)$ denotes the exponential function. 
Also note that $\log(\cdot)$ is analytic on $\mathbb{C}\backslash\mathbb{R}_{\le 0}$.

The following result establishes the equivalence between the homogeneous and heterogeneous cases, thereby reducing 
the instability analysis of the entire system to an analysis of its individual components.

\begin{prop}
\label{prop:hetro_eq_homo} 
Consider the cyclic network defined in Section~\ref{subsec:CycricProblem}. Suppose 
${\Delta}:=\diag({\delta}_1,\cdots,{\delta}_n) \in\mathbf{\Delta}_d$, 
with $\|{\Delta}\|_{\infty}{\ti=:}r < 1$ such that the cyclic network $\Sigma({\Delta},H,A)$ is stable. Then 
one can find ${\delta}\in\RHinf$ with 
$\|{\delta}\|_{\infty} \le \|{\Delta}\|_{\infty} +\epsilon$, with $\epsilon>0$ arbitrarily small, 
such that $\Sigma({\delta} I,H,A)$ is stable.
\end{prop}

\begin{proof}
The proof can be found in Appendix~\ref{subsec:proof_hetro_eq_homo}.
\end{proof}
{\ti
The result shows that if the network can be stabilized by a heterogeneous perturbation 
$\Delta\in\bfDel_d$, then it can also be stabilized by a homogeneous perturbation
$\delta I\in\bfDel_h$ with a slightly larger norm. The core idea is to construct, for each
$s$ in the CRHP, a homogeneous perturbation $\delta_{\rm av}$ such that
\begin{equation} \label{ave}
(1+\delta_{\rm av})^n=\prod_{i=1}^n(1+\delta_i).
\end{equation}
We then show that $|\delta_{\rm av}|\leq r$ whenever $|\delta_i|\leq r<1$
for all $i\in\II_n$. This follows from the convexity of the set 
\begin{align*} %
\bfThe:=\log(1+\bfdel), \hs
\bfdel:=\{\delta\in\IC:\,|\delta|\leq r\,\},
\end{align*}
as established in Lemma~\ref{lem:1} in the appendix. Finally,
the ``average'' $\delta_{\rm av}(s)$ is approximated by a real-rational
function, which might slightly increase the norm.

}

\begin{remark}
The same result holds for the case of feedback perturbation, where 
the perturbed system is represented by $h(s)/(1 + \delta_i(s))$. 
The proof of Proposition~\ref{prop:hetro_eq_homo} hinges on the convexity of the set $\mathbf{\Theta}$. 
Because the set $\mathbf{\Theta}$ corresponding to the feedback perturbation case remains unchanged, this crucial 
convexity property still holds. Consequently, the same proof directly applies to the case of feedback perturbation.
\end{remark}

Proposition~\ref{prop:hetro_eq_homo} leads to the following main result of this section. It shows that, 
in terms of finding a smallest distributed perturbation $\Delta$ that stabilizes
the nominal cyclic network $\ti G:=\mathcal{F}_{\ell}(H,A)$, 
heterogeneity will not lead to a better outcome.

\begin{theorem}
\label{thm:homo_equiv_hetero}
For the cyclic network $\Sigma(\Delta,H,A)$ defined above, we have %
$\rho_*(G) = \rho_h(G) \geq \rhoL(G)$, where
\begin{equation} \label{varrho}
\rhoL(G):=\max_{k\in\IU}\frac{1}{\|g_k\|_{\infty}}=\frac{1}{\mu}\max_{k\in\IU}\left(\inf_{\omega\in\IR}|\phi(j\omega)-\lambda_k|\right)
\end{equation}
and $g_k$ and $\IU$ are defined in and below~\eqref{def:gk}.
\end{theorem}

\begin{proof}
By Proposition~\ref{prop:hetro_eq_homo}, we have
\begin{align*}
\rho_h(G) :=\inf_{\delta\in \mathbb{S}_h(G)}\|\delta\|_{\infty} %
\le \inf_{\Delta\in \mathbb{S}_d(G)}\|\Delta\|_{\infty}:=\rho_*(G). 
\end{align*}
On the other hand, $\rho_*(G) \le \rho_h(G)$ by definition, and 
it is established in Theorem 1 of~\cite{NetworkRIR_short} that $\rhoL(G) \le \rho_h(G)${\ti.}
These facts lead to the claimed relation.
\end{proof}

{\ti

\begin{remark}
The lower bound $\rhoL$ on the homogeneous RIR $\rho_h$ was derived in
\cite{NetworkRIR_short} for general networks with diagonalizable connectivity
matrix $A$. The result was based on the observation that the characteristic
equation of the network with homogeneous perturbation $\delta I\in\bfDel_h$
is given by $g_k(s)\delta(s)=1$ for $k\in\II_n$ when the spectral decomposition
of $A$ is exploited. 
The equivalence between heterogeneous and homogeneous perturbations for the cyclic 
network (Proposition~\ref{prop:hetro_eq_homo}) then ensures that $\rhoL$ is also a 
lower bound on $\rho_*$. 
\end{remark}

}

\section{Cyclic Network: Exact RIR Analysis}
\label{sec:CycricExactRIR}
{\blue 
In this section, we will characterize a class of {\ti nominal agent $h(s)$ for the} cyclic network systems {\ti for which} %
the lower bound $\rhoL(G)$
{\ti in Theorem~\ref{thm:homo_equiv_hetero} is tight and} serves as the exact robust instability radius. Specifically, for any network system in 
this class, there exists a stable perturbation $\delta$ with $\|\delta\|_\infty = \rhoL(G)$ that renders the network 
{\ti with $\Delta=\delta I$} marginally stable. 
}

\subsection{Objective and Illustrate Example}

{\blue

In this subsection, we use an example to motivate the measure $\rhoL$ in~\eqref{varrho}---which
leads to tight RIRs of cyclic networks---and outline the procedure for deriving this tight RIR result. The transfer 
functions and parameters presented below are derived from the biological application studied later in Section~\ref{sec:BioAppli}. 
Readers are referred to Section~\ref{sec:BioAppli} for their physical meanings.

Consider a cyclic network of $n$ uncertain agents with dynamics given by $-\mu_i(1+\delta_i(s))h(s)$, 
where $n$ is odd and $\mu_i = \mu >0$ for all $i\in\II_n$. The nominal transfer function $h(s)$ is defined as 
\begin{align}
&  h(s) = \frac{c \beta}{(s+a)(s+b)},
  \label{eq:h-grn}\\
&  a = 5.1986, \ b = 0.4612, \ c = 30, \ \beta = 24 \label{eq:param-grn}
\end{align}
with $\mu= 0.006247$. 
This setup corresponds to the practical gene regulatory network application studied in~\cite{Niederholtmeyer2015}.  
}
\begin{figure}[tb]
\centering    %
\begin{subfigure}{0.47\columnwidth}
    \centering
    \includegraphics[width=\textwidth,clip]{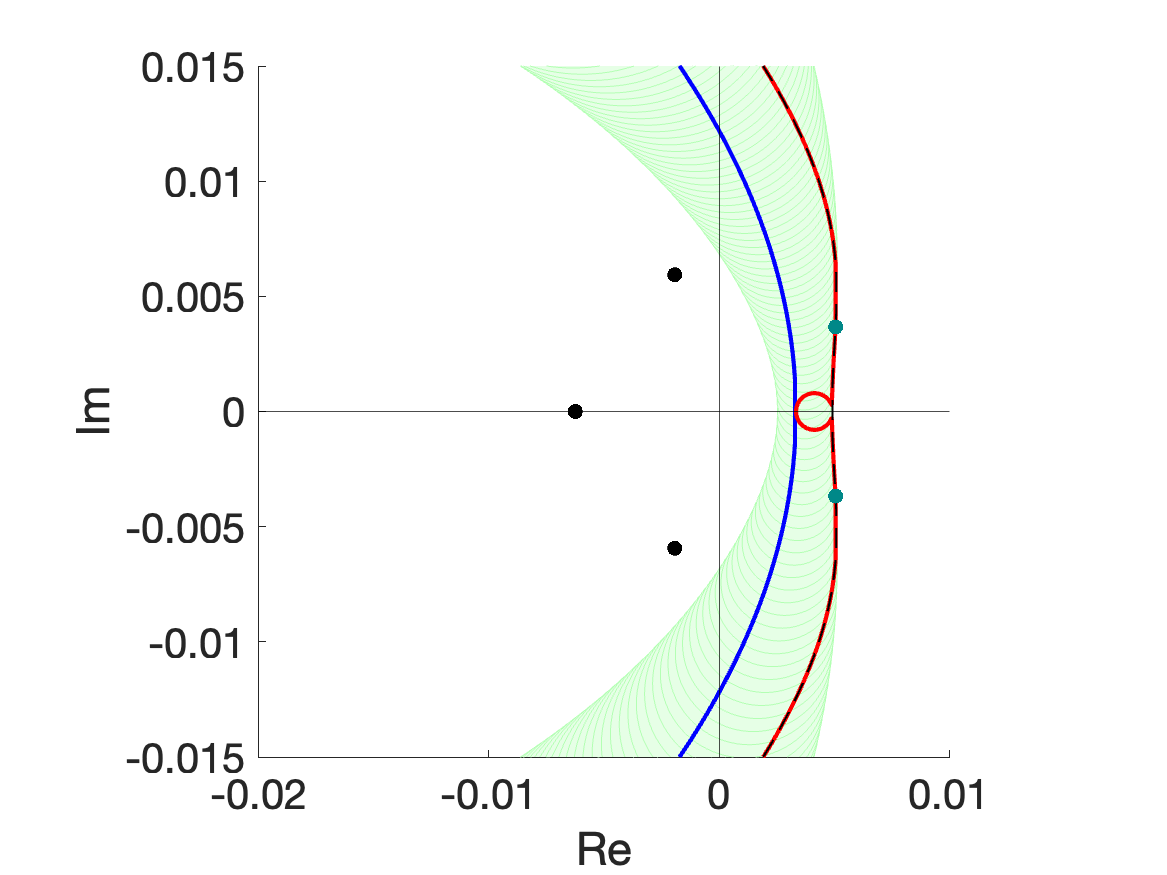}
    \caption{$n=5$}
    \label{fig:a}
\end{subfigure}
\begin{subfigure}{0.47\columnwidth}
    \centering
    \includegraphics[width=\textwidth,clip]{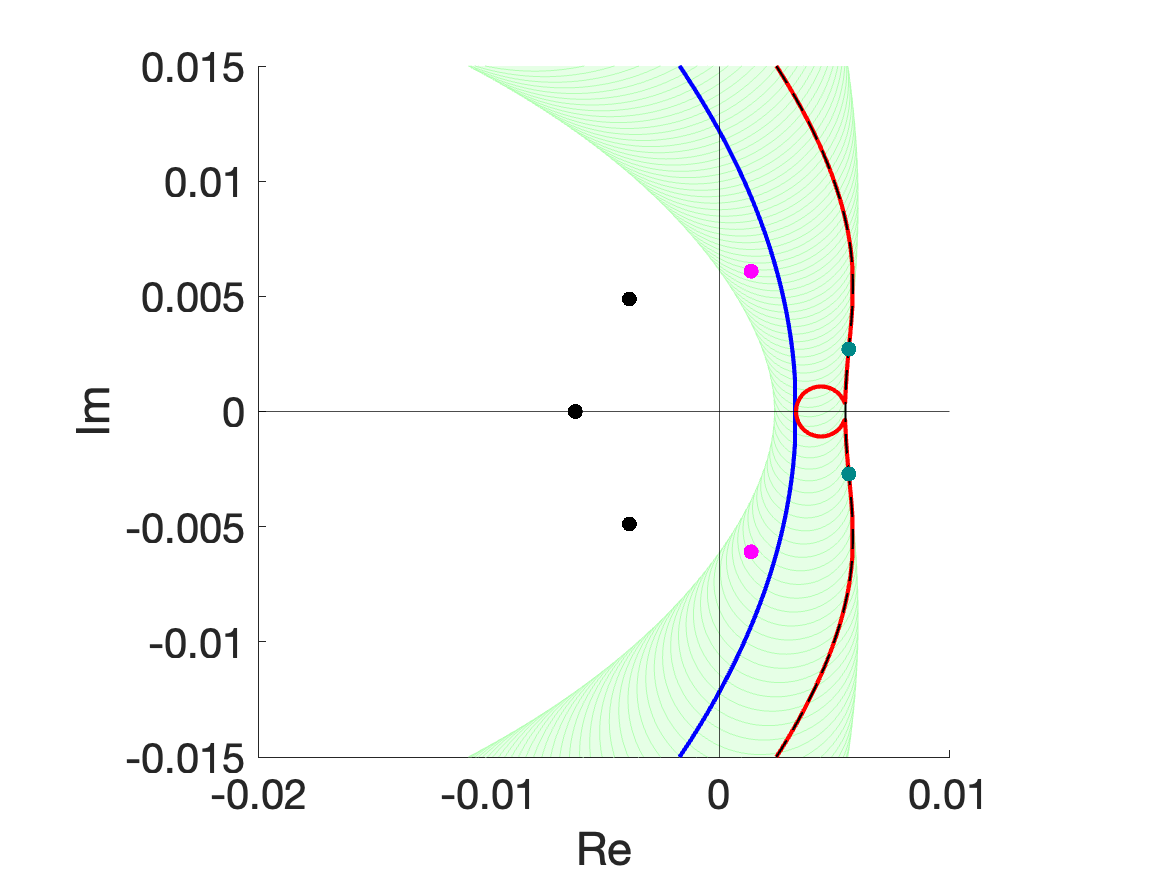}
    \caption{$n=7$}
    \label{fig:b}
\end{subfigure}

\begin{subfigure}{0.47\columnwidth}
    \centering
    \includegraphics[width=\textwidth,clip]{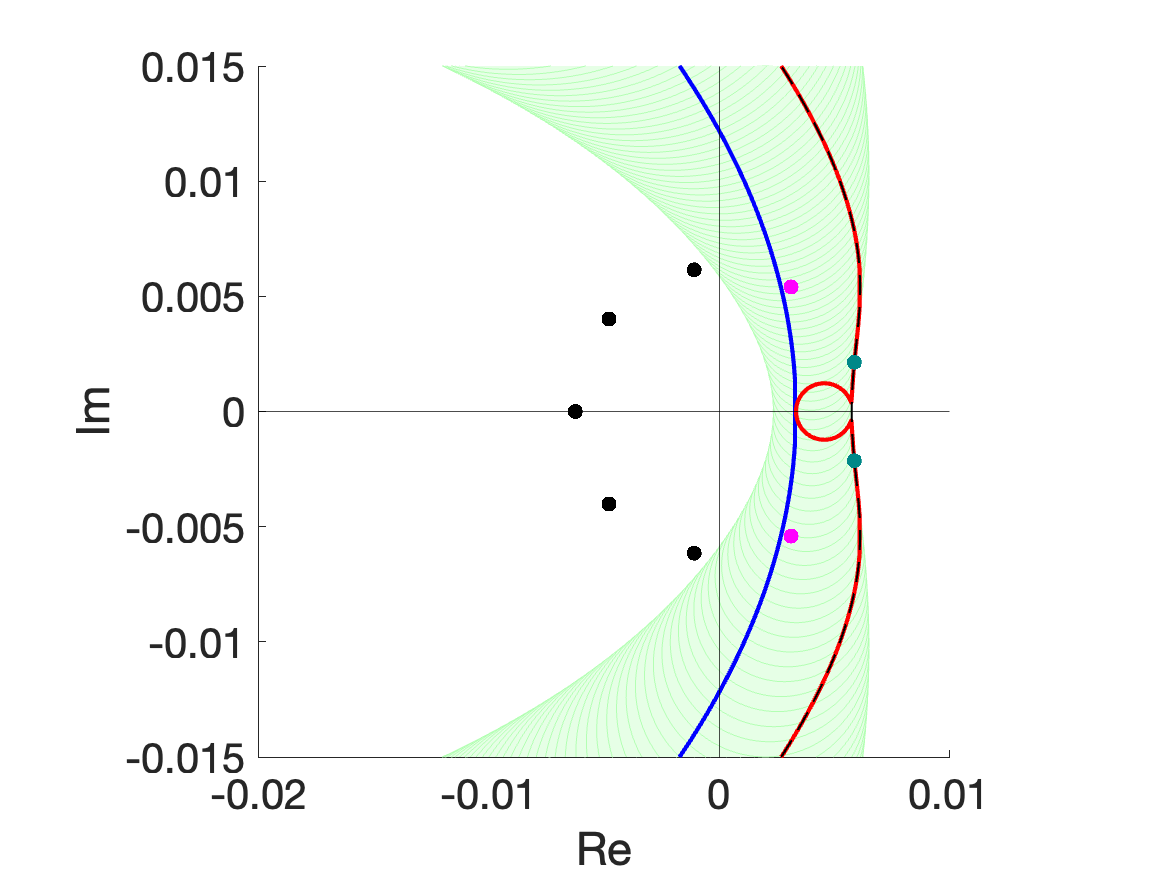}
    \caption{$n=9$}
    \label{fig:c}
\end{subfigure}
\begin{subfigure}{0.47\columnwidth}
    \centering
    \includegraphics[width=\textwidth,clip]{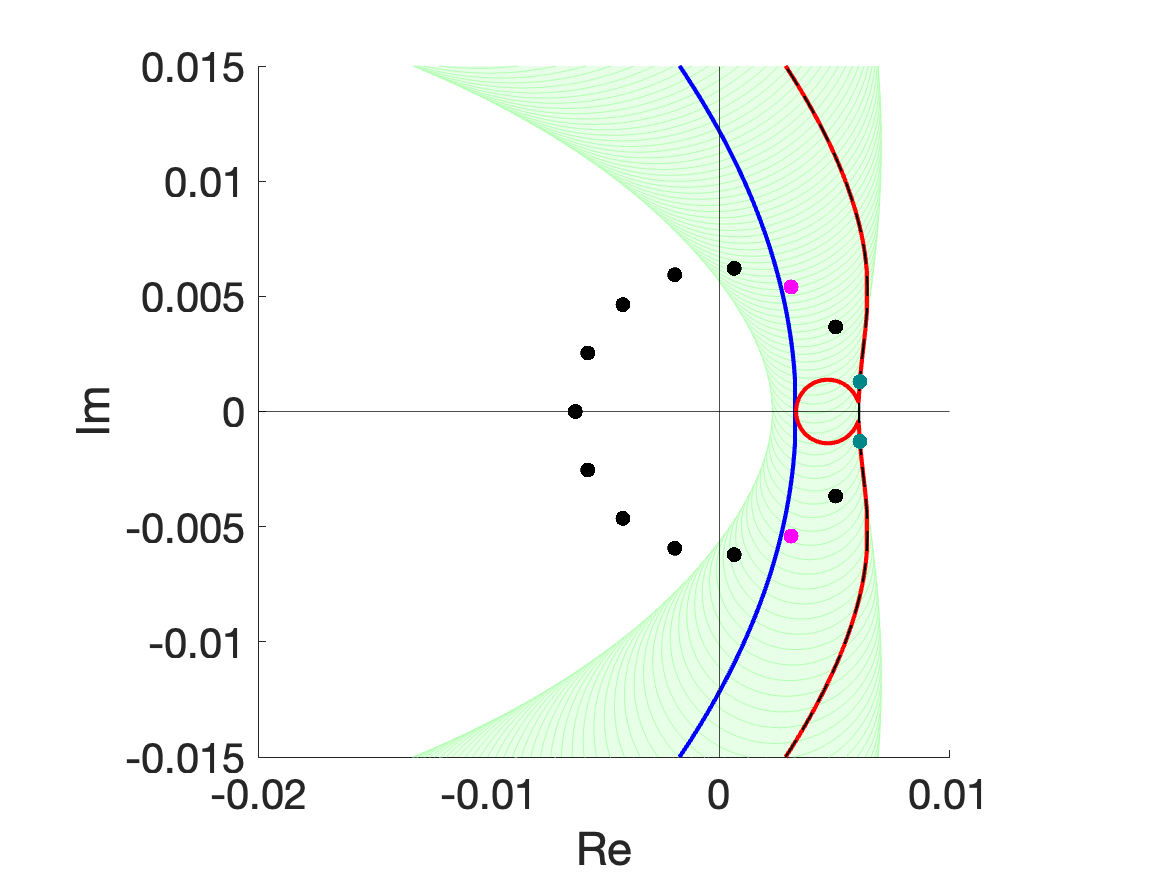}
    \caption{$n=15$}
    \label{fig:d}
\end{subfigure}

\caption{The stability region of $h(s)$ in the complex plain and the eigenvalues of $A$.}
\label{fig:four_images}
\end{figure}

{\blue
By the small-gain theorem, $\varrho_p(G):=\|G\|_{\infty}^{-1}$ is a lower bound on the RIR $\rho_*(G)$, where $\varrho_p(G)$ can be 
characterized as 
\begin{align}
\label{def:rhoP}
\varrho_p(G) = \min_{k\in\II_n}\frac{1}{\|g_k\|_{\infty}}
=\frac{1}{\mu}\min_{k\in\II_n}\left(\inf_{\omega\in\IR}|\phi(j\omega)-\lambda_k|\right).
\end{align}
On the other hand, Theorem~\ref{thm:homo_equiv_hetero} provides another lower bound, $\rhoL(G)$, on $\rho_*(G)$. By~\eqref{varrho} and~\eqref{def:rhoP}, 
that $\varrho_p(G) \leq \rhoL(G)$ clearly holds in general. 
For this particular example, the values of $\varrho_p(G)$ and $\rhoL(G)$ across a selection of the number of agents $n$ are summarized 
in the table below. 

\begin{table}[H]
\centering
\begin{tabular}{c|c|c|c|c|c|c|c}
$n$ & $3$ & $5$ & $7$ & $9$  & $11$ & $13$ & $15$ \\
\hline 
$\varrho_p$ & $0.069$ & $0.319$ & $0.172$ & $0.069$  & $0.197$ & $0.055$ & $0.069$ \\
\hline 
$\rhoL$    & $0.069$ & $0.319$ & $0.391$ & $0.420$ & $0.435$ & $0.444$ & $0.449$  \\
\end{tabular}
\end{table}

\noindent %
Except for $n=3$ and $n=5$, there are significant gaps between $\varrho_p(G)$ and $\rhoL(G)$, making $\rhoL(G)$ a much tighter lower bound on $\rho_*(G)$. 
In the following, we will explain the underlying reasons for this phenomenon. 

Recall that the stability of the nominal cyclic network $G:=\mathcal{F}_{\ell}(H,A)$ is equivalent to the stability 
of all $g_k$, $k\in\II_n$, where each $g_k$ is characterized as in~\eqref{def:gk}. By~\eqref{def:gk} and~\eqref{eq:h-grn}, 
one can readily verify that the stability of $g_k$ is characterized by having all eigenvalues $\lambda_1,\dots,\lambda_n$ of $A$ 
lie to the left side of the Nyquist plot of $\phi:=1/h$; any eigenvalue located on the right side implies instability. Eigenvalues
of $A$ located on the Nyquist plot of $\phi$ imply that the network has poles on the imaginary axis. 
In the case of 
the perturbed network with $\Delta:=\delta I_n$, we simply replace $g_k$ by $\tilde{g}_k:=\lambda_k/({\phi^{\dagger}}-\lambda_k)$, where ${\phi^{\dagger}}:=\phi/(1+\delta)$, 
and the same stability criterion applies. Thus, stabilizing an unstable network $G$ via the perturbation $\delta$ can be understood as 
finding a $\delta$ that modifies $\phi$ to ${\phi^{\dagger}}$ so that all eigenvalues of $A$ lie to the left of the Nyquist plot of ${\phi^{\dagger}}$. 
This interpretation is illustrated by Fig.~\ref{fig:four_images}. 

Figures~\ref{fig:four_images}(a)–(d) depict the scenarios for $n = 5, 7, 9,$ and $15$, respectively. In each plot, the blue curve represents the Nyquist plot of $\phi$,
while the $n$ dots on the circle of radius $\mu$ represent the eigenvalues of $A$. In all cases, several eigenvalues lie to the right of the blue curve, indicating that
the nominal network $G$ is unstable. The green region in each figure represents the following set:
\begin{align}
\label{the_set}
  \{\phi(j\omega)/(1+p) \mid \omega \in \mathbb{R}, \ p \in \mathbb{C}, \ |p| \le \rhoL \}, 
\end{align}
which specifies the admissible region for the Nyquist plot of ${\phi^{\dagger}}$ corresponding to any $\delta$ satisfying $\|\delta\|_{\infty}\le\rhoL$.  
Crucially, the green region covers those eigenvalues on the right, with the green-colored eigenvalues on the boundary of the region. 
This indicates that one might find a $\delta$ with $\|\delta\|_{\infty}\le\rhoL$ that marginally stabilizes 
the network $G$. One such $\delta$ is found, yielding a ${\phi^{\dagger}}$ whose Nyquist plot is depicted by the red curve. 
Note that by replacing $\rhoL$ in~\eqref{the_set} with a smaller number will shrink the green region, 
and thus it misses the green-colored unstable eigenvalues. This means, one would not be able to stabilize $G$ via a $\delta$ with norm smaller than $\rhoL$,
suggesting it is a tight lower bound of $\rho_*(G)$. 

Regarding the gap between $\varrho_p(G)$ and $\rhoL(G)$, note that according to~\eqref{varrho}, the lower bound $\rhoL(G)$ is determined by the 
green-colored eigenvalue(s) on the right, located farthest from the blue curve; this maximal distance yields the value of $\rhoL(G)$. In contrast, 
by~\eqref{def:rhoP}, the lower bound $\varrho_p(G)$ is determined by the eigenvalue(s) closest to the blue curve, regardless of which side of the curve they lie, 
with that minimal distance defining $\varrho_p(G)$. These closest eigenvalues are highlighted in red for $n=7$, $9$ and $15$. The figures clearly illustrate 
a significant gap between $\varrho_p(G)$ and $\rhoL(G)$ for these cases. For $n=5$, however, the two underlying eigenvalues happen to coincide, which explains 
why $\varrho_p(G)=\rhoL(G)$ in this case. 

It is worth noting that in this example, $\phi$ is a second-order polynomial whose Nyquist plot enjoys favorable properties, such as gain/phase monotonicity and 
convexity of the stability region. These properties play a crucial role in ensuring that $\rho_*(G)=\rhoL(G)$, and they motivate the development of this section 
as follows. 

First, the gain/phase monotonicity and convexity properties that underpin the tight characterization of $\rho_*(G)$ are introduced in Section~\ref{subsec:property_pgc}. Based on these properties, we will:
\begin{itemize}
    \item show in Section~\ref{subsec:crit_eig} that the critical eigenvalue giving rise to $\rhoL$ is $\lambda_*$, as defined in~\eqref{def:lambda_star};
    \item show in Section~\ref{subsec:PCR} that $g_*:=\lambda_*h/(1-\lambda_*h)$ can be (marginally) stabilized by a stable $\delta$ with $\| \delta \|_\infty = \rhoL$;
    \item show in Section~\ref{subsec:ClassesExactRIR} that this same $\delta$ simultaneously stabilizes all other systems $g_k(s)$, thereby implying that the structured perturbation $\Delta := \delta I_n$ renders the entire network (marginally) stable.
\end{itemize}
Finally, the characterization of a class of nominal agent dynamics $h$ having the required properties is presented in Section~\ref{subsec:class_of_phi}. While this class is by no means exhaustive, it nonetheless encompasses important models found in biological applications and explored in Section~\ref{sec:BioAppli}.
}

\subsection{Monotonicity and Convexity}
\label{subsec:property_pgc}

{\ti This section defines some properties for} a function $\phi:j\IR\mapsto\IC$, 
{\ti which represents the nominal agent dynamics $\phi(s):=1/h(s)$ in later sections.
L}et $\phi_R(\omega)$ and $\phi_I(\omega)$ be the real and imaginary parts of 
$\phi(j\omega)$; i.e., $\phi(j\omega) = \phi_R(\omega)+ j \phi_I(\omega)$. 
The phase and gain functions of $\phi$ are respectively %
{\ti defined as $\theta_\phi(\omega)$ and $r_\phi(\omega)$ such that
$\phi(j\omega)=r_\phi(\omega)e^{j\theta_\phi(\omega)}$.}
We assume
$\phi(0)>0$ and $\lim_{\omega\to\infty}r_{\phi}(\omega)=\infty$. Furthermore, if $\theta_{\phi}(\omega)=\pi$
for some frequencies, let $\omega_{\pi}$ {\ti denote} the smallest positive frequency among these; otherwise, 
$\omega_{\pi}$ is defined as $\infty$.  

In the {\ti following}, we use the superscripts $(\bullet)'$ and $(\bullet)''$ to
denote respectively the first and second derivatives of the function $(\bullet)$ 
{\ti with respect to (w.r.t.)} the frequency symbol $\omega$.

\begin{defn} 
A function $\phi:j\IR\mapsto\IC$ (with $\phi(0) > 0$) is said to be (low-frequency) phase-monotone increasing 
if 
\begin{align}
\theta_{\phi}(\omega_1) \leq \theta_{\phi}(\omega_2) ,  
\hs \forall \; 0 \leq \omega_1 < \omega_2 \leq \omega_\pi.  
\label{def:phase_mono_+}
\end{align} 
\end{defn}
\medskip
The class of phase-monotone increasing transfer functions is denoted by  ${\cal M}_{p}$. 
Note that for $\phi \in {\cal M}_{p}$, 
$\phi_R(\omega)\phi_I'(\omega)  \ge  \phi_I(\omega)\phi_R'(\omega)$
holds for all $\omega \in [0, \omega_\pi]$.

\smallskip

\begin{defn}
A function $\phi:j\IR\mapsto\IC$ (with $\phi(0) > 0$) is said to be (low-frequency) gain-monotone increasing if 
\begin{align} 
& r_{\phi}(\omega_1) \leq r_{\phi}(\omega_2), \ \ \forall \; 0 \leq \omega_1 < \omega_2 \leq \omega_\pi; \\
& r_{\phi}(\omega)   \geq r_{\phi}(\omega_{\pi}), \ \  \forall \; \omega > \omega_\pi.
\label{def:gain_mono_+}
\end{align}
\end{defn}
\medskip
The class of gain-monotone increasing transfer functions is denoted by  ${\cal M}_{g}$. 
Note that for $\phi \in {\cal M}_{g}$, 
$\phi_R(\omega)\phi_R'(\omega)  +  \phi_I(\omega)\phi_I'(\omega)  \ge 0$ 
holds for all $\omega \in [0, \omega_\pi]$.

Given $\phi \in \mathcal{M}_p \cap \mathcal{M}_g$, the curve $\phi(j\omega)$ over $\omega \in [0, \omega_\pi]$ is 
simple in the complex plane. Let $\mathcal{R}_\pi(\phi)$ denote the region of the complex plane bounded by 
$\phi(j\omega)$ for $\omega \in [-\omega_\pi, \omega_\pi]$. When $\omega_\pi$ is finite, this curve is simple and 
closed; hence, $\mathcal{R}_\pi(\phi)$ is a bounded, connected region containing the origin. When $\omega_\pi = 
\infty$, $\mathcal{R}_\pi(\phi)$ is a connected region lying to the left of $\phi(j\omega)$ for $\omega \in (-
\infty, +\infty)$, which again includes the origin. We now introduce a "convexity" property for $\phi$ based on 
$\mathcal{R}_{\pi}(\phi)$.
\begin{defn}
A function $\phi:j\IR\mapsto\IC$ (with $\phi(0) > 0$) is said to be convex if the corresponding 
$\mathcal{R}_{\pi}(\phi)$ region is convex. 
\end{defn}
\medskip
The class of ``convex'' transfer functions is denoted by  ${\cal M}_{c}$. Note that $\phi\in\mathcal{M}_c$
if and only if
\begin{align}
\begin{split}
&\phi_R'(\omega) \leq 0  \\
&\phi_R'(\omega)\phi_I''(\omega) 
- \phi_I'(\omega) \phi_R''(\omega)  \geq 0 
\end{split}
\label{eq:convex}
\end{align}
hold for all $\omega \in [-\omega_\pi, \omega_\pi]$.
To obtain the conditions in~\eqref{eq:convex}, one may consider $\phi_I$ to be an implicit
function of $\phi_R$ (when $\omega\in[0,\omega_{\pi}]$). This way, $\mathcal{R}_{\pi}(\phi)$ 
is convex if and only if $\phi_I(\phi_R)$ is a concave function, for which a necessary and sufficient condition is 
$\frac{d^2\phi_I}{d(\phi_R)^2} \le 0$. 
Existence of the implicit function in question requires $\phi_R'(\omega) \not = 0$ for $\omega\in(0,\omega_{\pi})$. 
Under this condition, the left-hand side of the above inequality is equal to
\begin{align*}
\frac{d}{d\omega}\left(\frac{\phi_I'}{\phi_R'} \right)\cdot\left(\phi_R'\right)^{-1}
=\frac{\phi_I''\phi_R'-\phi_R''\phi_I'}{(\phi_R')^3}{\ti.}
\end{align*}
Clearly, convexity of $\mathcal{R}_{\pi}(\phi)$ requires $\phi_R' < 0$ for $\omega\in (0,\omega_{\pi})$; thus
convexity of $\mathcal{R}_{\pi}(\phi)$ holds true if and only if conditions in~\eqref{eq:convex} are satisfied. 

\begin{remark}
One can readily verify the following equivalent inequalities
for gain-monotone, phase-monotone, and convex properties of $\phi$. Below we drop the $\omega$ dependency 
for simplicity; the inequalities are to hold for $\omega\in[0,\omega_{\pi}]$. 
\begin{align}
\label{ineq:pola_g}
&\mbox{gain-monotone increasing: \ \ } r_{\phi}' \ge 0,      \\
\label{ineq:pola_p}
&\mbox{phase-monotone increasing: \ \  } \theta_{\phi}'\ge 0,   \\
&\mbox{convexity:  } \notag \\
\label{ineq:pola_c1}
&\hspace{1cm} r_{\phi}'\cos\theta_{\phi} - r_{\phi}\sin\theta_{\phi}\theta_{\phi}' \le 0 \\
\label{ineq:pola_c2}
&\hspace{1cm} 
2(r_{\phi}')^2\theta_{\phi}' + r_{\phi}^2(\theta_{\phi}')^3 + r_{\phi}r_{\phi}'\theta_{\phi}''-r_{\phi}r_{\phi}''\theta_{\phi}' \ge 0
\end{align}
Note that the first two terms on the left-hand side of~\eqref{ineq:pola_c2} are positive. Hence, a sufficient condition
for inequality~\eqref{ineq:pola_c2} is $\theta_{\phi}''/\theta_{\phi}' \ge r_{\phi}''/r_{\phi}'$, 
assuming $\phi$ is gain- and phase-monotone increasing. 
\end{remark}
\subsection{Critical Unstable Eigenvalue}
\label{subsec:crit_eig}

Recall $\rhoL$ defined in~\eqref{varrho} and $g_k = \lambda_k/(\phi-\lambda_k)$ as described in~\eqref{def:gk},
where $\phi:=1/h$ and $\lambda_k$ is an eigenvalue of $A$. The following result characterizes a class of $\phi$ 
for which $\rhoL$ is obtained by $\lambda_*$ defined in~\eqref{def:lambda_star}.

\begin{prop}
\label{prop:cril_eig}
Consider the cyclic network $\Sigma(\Delta,H,A)$ with $n$ agents ($n\ge 2$), where the 
nominal network $\mathcal{F}_{\ell}(h\cdot I_n, A)$ is unstable and the transfer function $h$, which represents the nominal
agent dynamics, is strictly proper and stable, and satisfies $h(0)>1/\mu >0$. Suppose each unstable 
$g_k$ has the peak-gain frequency in $[0,\omega_{\pi})$, and 
$\phi:=1/h$ satisfies the following conditions:
\begin{itemize}
\item[] (\ref{prop:cril_eig}.1) $\phi(j\omega)$ is gain-monotone increasing; 
\item[] (\ref{prop:cril_eig}.2) $\mathcal{R}_{\pi}(\phi)$ is convex. 
\end{itemize}
Then we have 
\begin{align}
\rhoL := \max_{k\in\IU} \frac{1}{\|g_k\|_{\infty}} = \frac{1}{\|g_*\|_{\infty}}, \ \
g_*:=\frac{\lambda_*}{\phi-\lambda_*}.
\label{def:varrho_plus}
\end{align}
In other words, $\lambda_*$ is the farthest unstable eigenvalue from $\phi(j\omega)$. 
\end{prop}

\begin{proof} 
Note that $0<\phi(0)<\mu$, and without loss of generality, let $\mu=1$. Furthermore, 
note that for $g_k$ to be unstable, $\lambda_k$ must be outside $\mathcal{R}_{\pi}(\phi)$. 
Since the distribution of $\lambda_k$ and $\mathcal{R}_{\pi}(\phi)$ are symmetric w.r.t. the real axis, we {\ti may} %
only consider those in the upper half plane; i.e., 
$\ti\mathcal{R}_{\pi}^u(\phi):=\{z\in\mathcal{R}_\pi(\phi)~|~\Im(z)\geq0~\}$ 
and those $\lambda_k$ on the portion of upper half unit circle which is outside 
$\mathcal{R}_{\pi}^u(\phi)$. Denote this portion of upper half unit circle by 
$\mathcal{C}_{\phi}:=\{\mathrm{e}^{j\theta}~|~\theta\in[0,\theta_c]\}$, where 
\begin{align*} 
\theta_c:=\begin{cases}
\angle\phi(j\omega_c) & \mbox{if $\exists \ \omega_c\in[0,\omega_{\pi}]$ such that $|\phi(j\omega_c)|=1$;} \\
\pi & \mbox{otherwise}{\ti.}
\end{cases}
\end{align*}
In the second case, $\mathcal{C}_{\phi}$ is the entire upper half unit circle. 
 
Moreover, since 
\begin{align*}
&\|g_k\|_{\infty}=\sup_{\omega\in[0,\omega_{\pi})}1/|\phi(j\omega)-\lambda_k|\\ 
\Leftrightarrow  \ &1/\|g_k\|_{\infty}=\inf_{\omega\in[0,\omega_{\pi})}|\phi(j\omega)-\lambda_k|.
\end{align*}
it means $1/\|g_k\|_{\infty}$ is equal to the distance from $\lambda_k$ to the {\ti set} $\mathcal{R}_{\pi}^u(\phi)$. 
Thus, to prove that $\rhoL$ is obtained by $\ti g_*$, we will show that 
\begin{align}
\label{def:D_theta}
D(\theta):= \inf_{\omega\in[0,\omega_{\pi})}|\phi(j\omega)-\mathrm{e}^{j\theta}|, \ \theta\in[0,\theta_c]
\end{align}
is {\ti a} \emph{decreasing} function; thus, $\lambda_*$, which has the smallest angle, maximizes the distance.

Given a point $\mathrm{e}^{j\theta}\in\mathcal{C}_{\phi}$, since $\mathcal{R}_{\pi}(\phi)$ is convex, one can find a unique point
on the boundary of $\mathcal{R}_{\pi}^u(\phi)$ that gives rise to $D(\theta)$. Let the point be $\phi(j\omega^*):=\phi_R(\omega^*)+j\phi_I(\omega^*)$; 
note that $\omega^*$ is a function of $\theta$, and here we drop the $\theta$-dependency for notational simplicity. Moreover, since
$\phi(j\omega^*)$ is the projection of $\mathrm{e}^{j\theta}$ on the boundary of $\mathcal{R}_{\pi}(\phi)$, we must have {\ti the orthogonality}
\begin{align}
\label{eq:proj_cond}
\hspace{-0.3cm}
(\cos\theta - \phi_R(\omega^*))\phi_R'(\omega^*)+(\sin\theta - \phi_I(\omega^*))\phi_I'(\omega^*)=0{\ti.}
\end{align}
In light of~\eqref{eq:proj_cond}, we have for some real number $\beta$, 
\begin{align}
\begin{split}
&\cos\theta - \phi_R(\omega^*) = \beta \phi_I'(\omega^*), \\ 
&\sin\theta - \phi_I(\omega^*) =-\beta \phi_R'(\omega^*){\ti.} 
\end{split}
\label{eq:perp_rela2}
\end{align}
Note that $D(\theta)^2 = \left(\cos\theta - \phi_R(\omega^*) \right)^2 + \left(\sin\theta - \phi_I(\omega^*)\right)^2$. 
Differentiating $D$ w.r.t. $\theta$, we obtain
\begin{align}
&D(\theta)D'(\theta) =
\left(\cos\theta - \phi_R(\omega^*)\right)\left(-\sin\theta - \phi_R'(\omega^*)\frac{d\omega^*}{d\theta}\right) 
\notag \\
&\hspace{2cm} + \left(\sin\theta - \phi_I(\omega^*)\right)\left(\cos\theta - \phi_I'(\omega^*)\frac{d\omega^*}{d\theta}\right) 
\notag \\
& = \left(\cos\theta- \phi_R(\omega^*)\right)(-\sin\theta)+\left(\sin(\theta)-\phi_I(\omega^*)\right)\cos\theta \notag \\
& = -\beta \left(\phi_I'(\omega^*)\sin\theta + \phi_R'(\omega^*)\cos\theta \right){\ti.} \label{eq:ddistance}
\end{align}
Here we apply~\eqref{eq:proj_cond} for the second equality, and~\eqref{eq:perp_rela2} for the third. 

Now since $\mathcal{R}_{\pi}(\phi)$ is convex, {\red we must have $\phi_R'(\omega^*)<0$ and $\sin\theta - \phi_I(\omega^*)>0$. 
} 
These imply that $\beta >0$ by the second equation in~\eqref{eq:perp_rela2}. 
Applying~\eqref{eq:perp_rela2} again to~\eqref{eq:ddistance}, we obtain
\begin{align}
&\hspace{0.2cm}D(\theta)D'(\theta) = -\beta(\phi_R'(\omega^*)\cos\theta+\phi_I'(\omega^*)\sin\theta) 
\notag \\
&= -\beta\left(\phi_R(\omega^*)\phi_R'(\omega^*)+\phi_I(\omega^*)\phi_I'(\omega^*)\right) <0{\ti.}
\label{ineq:ddistance}
\end{align}
The last inequality follows from the assumption that $\phi$ is gain-monotone increasing. Thus we have shown  
$D'(\theta)<0$, which concludes the proof. 
\end{proof}

\subsection{PCR Condition}
\label{subsec:PCR}
In this subsection, we characterize the properties of $\phi$ under which
the subsystem $g_*$ described in~\eqref{def:varrho_plus} satisfies 
the required PCR condition in Theorem~\ref{thm:MS} at its peak-gain frequency. 

\begin{prop} 
\label{prop:gen_pcr}
Let $\phi$ be a rational transfer function and $\lambda\in\mathbb{C}$ such that 
$g:=\lambda/(\phi-\lambda)\in\mathcal{G}_1^{\omega_p}$.  Suppose
$\phi$ is gain- and phase-{\ti monotone increasing}. Then the following properties hold: 
\begin{itemize}
\item $\theta_g'(\omega_p) > 0$; 
\item if $\omega_p\not=0$, $\theta_g'(\omega_p) > |\sin\theta_g(\omega_p)/\omega_p|$ holds provided
$\phi$ satisfies
\begin{align}
\label{ineq:gen_pcr}
\left.|\phi|(|\phi|)'\right|_{\omega = \omega_p} \le \left.\omega_p|\phi'|^2\right|_{\omega=\omega_p}, 
\end{align}
where $(|\phi|)':=\frac{d|\phi(j\omega)|}{d\omega}=(\phi_R\phi_R'+\phi_I\phi_I')/|\phi|$. 
\end{itemize}
Consequently, $g$ can be marginally stabilized by $\delta$ of the form
$\delta(s)\equiv 1/\|g\|_{\infty}$ (when $\omega_p=0$), or 
\begin{align}
\label{def:delta_star}
\delta(s)=\frac{1}{\|g\|_{\infty}}\left(\frac{s-a}{s+a}\right) \mbox{ for some $a>0$,}
\end{align}
(when $\omega_p\not=0$), and hence %
{\ti $\rho_*(g)=1/\|g\|_\infty$.}
\end{prop}

\begin{proof}
Let $\lambda = \lambda_R+j\lambda_I${\ti.} In the {\ti following}, we drop the $\omega$-dependency for notational simplicity. 

Note that at the peak-gain frequency $\omega_p$, $g$ satisfies $\frac{d}{d\omega}\log|g(j\omega)|=0$, and therefore we have
\begin{align*}
\mathrm{Re}\left(\frac{g'}{g}\right)=\mathrm{Re}\left(\frac{-\phi'}{\phi-\lambda}\right) = 0 \ \Leftrightarrow \
\mathrm{Re}\left(\phi'(\overline{\phi-\lambda})\right) = 0, 
\end{align*}
which leads to 
\begin{align}
\phi_R'(\phi_R-\lambda_R)+\phi_I'(\phi_I-\lambda_I) = 0, 
\label{eq:pg_cond}
\end{align} 
at the peak-gain frequency $\omega_p$. By~\eqref{eq:pg_cond}, we have
\begin{align}
\begin{bmatrix}\phi_R-\lambda_R \\ \phi_I-\lambda_I  \end{bmatrix} = \frac{1}{\alpha}\begin{bmatrix} -\phi_I' \\ \phi_R' \end{bmatrix}
\ \Leftrightarrow \ \lambda = \phi - \frac{j}{\alpha} \phi',
\label{expr1_lambda}
\end{align}
where $\alpha$ is a positive real number. Note that positivity of $\alpha$ is a result of $\lambda$ being outside $\mathcal{R}_{\pi}(\phi)$,
which is necessary for $g$ to be unstable. 

Moreover, apply{\ti ing}~\eqref{expr1_lambda}{\ti} yields at $\omega_p$, 
\begin{align}
g &= \frac{\lambda}{\phi-\lambda}
= \frac{(\alpha\phi_R+\phi_I')+j(\alpha\phi_I-\phi_R')}{-\phi_I'+j\phi_R'} =\frac{g_R + j g_I}{|\phi'|^2} 
\end{align} 
where %
\begin{align}
g_R &= -\phi_I'(\alpha\phi_R+\phi_I')+\phi_R'(\alpha\phi_I-\phi_R')\notag \\ 
 &=-|\phi'|^2-\alpha\left(\phi_I'\phi_R-\phi_R'\phi_I\right) \label{eq:gR}\\ 
g_I &= -\alpha(\phi_R'\phi_R+\phi_I\phi_I') = -\alpha|\phi|(|\phi|)'{\ti.}  \label{eq:gI}
\end{align}
Applying~\eqref{expr1_lambda} to~\eqref{eq:gI}, we find, at $\omega_p$, 
\begin{align}
&\theta_g' = \mathrm{Im}\left(g'/g\right)=\mathrm{Im}\left(-\phi'/(\phi-\lambda)\right) = \alpha, \label{eq:pcr_g} \\
&\sin\theta_g = g_I/(g_R^2+g_I^2)^{\frac{1}{2}} = -\alpha |\phi|(|\phi|)' /(g_R^2+g_I^2)^{\frac{1}{2}}{\ti.} \label{eq:sin_t_g}
\end{align}
Thus, by~\eqref{eq:pcr_g} we have $\theta_g'(\omega_p)>0$ regardless of the value of $\omega_p$. 
Furthermore, (for $\omega_p\not = 0$) by~\eqref{eq:sin_t_g}, 
$\theta_g' \ge |\sin\theta_g/\omega_p|$ is equivalent to $1\ge \left|\frac{|\phi|(|\phi|)' }{\omega_p\sqrt{g_R^2+g_I^2}}\right|$, 
which holds iff %
\begin{align}
&|\phi|^2((|\phi|)')^2 \le \omega_p^2(g_R^2+g_I^2) \label{eq:eqiv_cond1} \\
&=\omega_p^2\left(\left(|\phi'|^2+ \alpha(\phi_I'\phi_R-\phi_R'\phi_I) \right)^2
+\alpha^2|\phi|^2((|\phi|)')^2\right) \notag
\end{align}
Now, that $\phi$ is phase-monotone increasing implies $\phi_I'\phi_R-\phi_R'\phi_I \ge 0$, and one can readily verify that, w.r.t. $\alpha$, 
the right-hand side of~\eqref{eq:eqiv_cond1} is minimized when $\alpha=0$. Setting $\alpha =0$, \eqref{eq:eqiv_cond1} reduces to 
\begin{align}
|\phi|^2((|\phi|)')^2 \le \omega_p^2 |\phi'|^4
\label{eq:eqiv_cond2}
\end{align}
which is equivalent to~\eqref{ineq:gen_pcr} should $\phi$ be also gain-monotone increasing. 

Finally, that $g$ can be marginally stabilized by $\delta$ in the form of~\eqref{def:delta_star} and hence %
{\ti $\rho_*(g)=1/\|g\|_\infty$} follows Theorem~\ref{thm:MS} (and the readers are referred to~\cite[Theorem 1]{hara2023exact} for more details).
\end{proof}

\begin{remark}
Since $\omega_p$ depends on $\lambda$, {\ti inequality}~\eqref{ineq:gen_pcr} is $\lambda$-dependent. A more conservative condition independent 
of $\lambda$ is to require that~\eqref{ineq:gen_pcr} holds for all $\omega \in [0, \infty)$, or for $\omega$ within a specific frequency range 
(e.g., $[0, \omega_{\pi}]$) containing $\omega_p$. As a result, the condition solely involves $\phi$.
\end{remark}

\begin{remark}
Let $\phi = r_{\phi}\mathrm{e}^{j\theta_{\phi}}$. One can verify that inequality~\eqref{ineq:gen_pcr} can also be expressed 
in polar coordinates as
\begin{align}
\left. r_{\phi} r_{\phi}' \right|_{\omega=\omega_p} \le \omega_p \left.\left((r_{\phi}')^2 + r_{\phi}^2(\theta_{\phi}')^2 \right)\right|_{\omega=\omega_p}
\label{ineq:eqiv_cond3}
\end{align}
This expression proves useful in the subsequent development.
\end{remark}

\subsection{Simultaneous Stabilization and Exact RIR}
\label{subsec:ClassesExactRIR}

Recall $g_k$, $k=1,\cdots, n$ defined in~\eqref{def:gk} and $\rhoL$ in~\eqref{varrho}. Suppose $\rhoL$ is equal to $1/\|g_*\|_{\infty}$, and  
$g_*\in\mathcal{G}_1^{\omega_p}$ satisfies the PCR condition in 
{\red case (I)-(ii) of Theorem~\ref{thm:MS}}. As such, $g_*$ has the exact {\ti RIR} equal to $\rhoL$; i.e., 
there {\ti is a} first-order all-pass function $\delta$ with $\|\delta\|_{\infty} = 1/\|g_*\|_{\infty} = \rhoL$ such that 
$[\![g_*, \delta]\!]$ is marginally stable. In the following, we propose a sufficient condition on $\phi$ such that 
$\delta$ would also (marginally) {\ti stabilize} all other unstable $g_k$'s {\ti and maintain stability of the remaining $g_k$'s.}

\begin{prop}
\label{prop:simul_stab}
Consider $g_*$ defined in %
{\ti\eqref{def:varrho_plus}} with 
a peak-gain frequency $\omega_p$, and let $\delta_*$ of the form~\eqref{def:delta_star}
be such that $[\![g_*, \delta_*]\!]$ is marginally stable. Define $F(\Omega):=|\phi(j\omega)|^2$, where $\Omega:=\omega^2$. 
If $\phi$ is gain-monotone increasing and $F$ satisfies 
\begin{align}
\label{ineq:cond_exactRIR}
F''(\Omega)/F'(\Omega) \geq  -2/(\Omega + a^2) 
\;\; \forall \Omega \geq 0,
\end{align}
where $F':=\frac{d F}{d\Omega} $ and $F'':=\frac{d^2 F}{d\Omega^2}$,
then $[\![g_k, \delta_*]\!]$ is (marginally) stable for $k\in\mathbb{U}$. 
\end{prop}

\begin{proof}
Note that stability (resp. marginal stability) of $[\![g_k,\delta_*]\!]$ is implied by 
$\lambda_k$ belonging to the interior (resp. boundary) of $\mathcal{R}_{\pi}({\phi^{\dagger}})$, where 
${\phi^{\dagger}} := \phi/(1+\delta_*)$.
Also note that $\delta_*$ must satisfy $g_*(j\omega_p)\delta_*(j\omega_p)=1$, which is equivalent to 
$\lambda_*(1+\delta_*(j\omega_p))=\phi(j\omega_p)$, or ${\phi^{\dagger}}(j\omega_p)=\lambda_*$. As such, $|{\phi^{\dagger}}(j\omega_p)| = |\lambda_*|$
and stability of all other $[\![g_k,\delta_*]\!]$ is obtained if one can show that 
$$
|{\phi^{\dagger}}(j\omega)|\ge |\lambda_*| \quad \forall \omega\ge \omega_p,
$$
which is implied by $(|{\phi^{\dagger}}(j\omega)|)'\ge 0$ for all $\omega\ge\omega_p$. To this end, define
\begin{align*}
M(\Omega):=\frac{1}{|1+\delta_*(j\omega)|^2} = \gamma \frac{\Omega+a^2}{\Omega+z^2},
\end{align*}
where $\gamma = \|g_*\|_{\infty}/(\|g_*\|_{\infty}+1)$ 
and $z = a(\|g_*\|_{\infty}-1)/(\|g_*\|_{\infty}+1)$. Observe that 
$F(\Omega)M(\Omega)/\gamma = |{\phi^{\dagger}}(j\omega)|^2/\gamma$, and 
\begin{align}
\frac{d}{d\Omega}\frac{F(\Omega)M(\Omega)}{\gamma}
&=\frac{(\Omega + a^2)(\Omega + z^2)F'(\Omega) - (a^2-z^2)F(\Omega)}{(\Omega + z^2)^2}\notag\\
&{\ti =:}\frac{H(\Omega)}{(\Omega + z^2)^2}.
\end{align}
Thus $(|{\phi^{\dagger}}(j\omega)|)'\ge 0$ is equivalent to $H(\Omega)\ge 0$. We will now proceed to show that
\begin{align}
H(\Omega)\ge 0, \ \ \forall \Omega\ge\Omega_p:=\omega_p^2{\ti.}
\label{eq:H}
\end{align}
Note that~\eqref{eq:H} is implied by the following conditions
\begin{subequations}
\label{eq:HC}
\begin{align}
& H(\Omega^*)\ge 0, \ \  \mbox{for some $\Omega^*\in [0,\Omega_p)$}{\ti,} \label{ineq:HC1}\\ 
& \frac{d}{d\Omega}H(\Omega)\ge 0, \ \forall \Omega\ge 0{\ti.} \label{ineq:HC2}
\end{align}
\end{subequations}
Observe that
\begin{align*}
\frac{d}{d\Omega} H(\Omega) =(\Omega+a^2)(\Omega+z^2)F''(\Omega) + 2(\Omega + z^2) F'(\Omega).
\end{align*}  
Since $F'(\Omega)$ is positive due to the gain monotonicity of $\phi(s)$, it is clear that
inequality~\eqref{ineq:HC2} is equivalent to~\eqref{ineq:cond_exactRIR}. 

Finally we note that~\eqref{ineq:HC1} is implied by $|{\phi^{\dagger}}(0)| < |\lambda_*|$. 
Since $|{\phi^{\dagger}}(\omega_p)| = |\lambda_*|$, $|{\phi^{\dagger}}(0)| < |\lambda_*|$ means 
that $(|{\phi^{\dagger}}(j\omega)|)' \ge 0$ must occur for some $\omega\in[0,\omega_p)$, which is equivalent to~\eqref{ineq:HC1}.
To see $|{\phi^{\dagger}}(0)| < |\lambda_*|$, suppose $|{\phi^{\dagger}}(0)| \ge |\lambda_*|$, and without loss of generality let $|\lambda_*|=1$. 
This implies $\phi(0) \ge 1-1/\|g_*\|_{\infty}$, which in turn implies 
\begin{align*}
|\phi(j\omega_p)-\lambda_*| \ge 1-\phi(0)
\ \Leftrightarrow \ D(\angle \lambda_*) \ge D(0),
\end{align*}
where $D(\theta):=\inf_{\omega\in[0,\omega_{\pi})}|\phi(j\omega)-\mathrm{e}^{j\theta}|$ is defined in~\eqref{def:D_theta}. As we have 
shown that $D(\theta)$ is a \emph{decreasing} function for $\theta\in[0,\angle\lambda_*]$, this contradicts $D(\angle \lambda_*) \ge D(0)$ and  
we may conclude that $|{\phi^{\dagger}}(0)| < |\lambda_*|$.
\end{proof}

\begin{remark} 
Since $F'(\Omega)$ is positive, \eqref{ineq:cond_exactRIR} would hold if $F''(\Omega)\ge 0$ 
for all $\Omega \ge 0$. Classes of $\phi$ satisfying the above condition would be given in the next subsection.
\end{remark}

Our main result about the exact RIR {\ti[?]} of a cyclic network is formulated as follows. The result concerns
the cyclic network where the nominal network $\LFT(h\cdot I_n, A)$ is unstable. The nominal agent dynamics
$h$ is stable, strictly proper, and satisfies $h(0)>1/\mu>0$.
\begin{theorem}
\label{thm:cyclic_RIR}
Suppose the inverse of the nominal agent dynamics, $\phi:=1/h$ satisfies the following 
conditions:
\begin{align}
&\mbox{$\phi$ is gain-monotone increasing;} \label{phi_cond1}\\
&\mbox{$\phi$ is phase-monotone increasing;} \label{phi_cond2}\\
&\mbox{$\phi$ is convex;} \label{phi_cond3}\\
&|\phi|(|\phi|)' \le  \omega |\phi'|^2, \ \forall \omega\in[0,\infty); \label{phi_cond4} \\
& F''(\Omega)>0, \ \forall \Omega > 0,  \label{phi_cond5}
\end{align}
wherein the last condition $\Omega:=\omega^2$ and $F(\Omega):=|\phi(j\omega)|^2$. Then the cyclic network has 
its {\ti RIR exactly} equal to $\rhoL:=1/\|g_*\|_{\infty}$, where $g_*=\lambda_*/(\phi-\lambda_*)$ and $\lambda_*$ is 
defined in~\eqref{def:lambda_star}.
\end{theorem}
\begin{proof}
The result is a direct consequence of Propositions~\ref{prop:cril_eig},~\ref{prop:gen_pcr}, and~\ref{prop:simul_stab}. 
Recall that 
\begin{itemize}
\item \eqref{phi_cond1} and~\eqref{phi_cond3} were utilized to establish the critical unstable 
eigenvalue $\lambda_*$ that characterizes $\rhoL$ (cf.~\eqref{def:varrho_plus}); 
\item \eqref{phi_cond1},~\eqref{phi_cond2}, and~\eqref{phi_cond4} were utilized to establish that
$g_*$ can be (marginally) stabilized by a stable $\delta_*$ with $\|\delta_*\|_{\infty} = 1/\|g_*\|_{\infty}$, 
and hence has the exact RIR; 
\item \eqref{phi_cond1} and~\eqref{phi_cond5} were utilized to establish that $\delta^*$ simultaneously stabilizes
all {\ti other} %
$g_k$'s, and hence the network. 
\end{itemize}
These results in turn lead to {\ti an exact characterization of} the RIR property for the cyclic network in question. 
\end{proof}

\subsection{Class of $\phi$ Resulting in Exact RIR}
\label{subsec:class_of_phi}

In the previous subsection, we have established that a cyclic network has {\ti the} RIR equal to
$\rhoL$ if the inverse nominal agent dynamics $\phi$ satisfies~\eqref{phi_cond1} to~\eqref{phi_cond5}. 
In this subsection, a class of $\phi$ is shown to possess these properties. The characteristics of this class
are easy to verify, and the class covers the dynamics of many practical systems, including those with over or heavily damped dynamics.  
 
\begin{theorem}
\label{thm:main_class_phi}
Let $\phi$ be a real Hurwitz polynomial and $\phi(0)>0$. Then $\phi$ is phase-monotone increasing and convex. Moreover, 
consider the following subset of Hurwitz polynomials:
\begin{align*}
\mathcal{H}_p&:={\ti\{}\phi(s)=\prod_{i=1}^n(s-z_i)~|~\mbox{$\phi$ is real}, \ \mathrm{Re}(z_i)<0,  \\ 
&\hspace{2cm} \left. |\mathrm{Im}(z_i)| \le |\mathrm{Re}(z_i)|, \ i=1,\cdots, n; \ n\in\mathbb{N} \right\}{\ti.}
\end{align*}
If $\phi\in\mathcal{H}_p$, then $\phi$ is also gain-monotone 
increasing %
and satisfies conditions~\eqref{phi_cond4} and~\eqref{phi_cond5}. 
\end{theorem}

Before presenting the proof, we state the following technical lemmas as they are instrumental in proving the theorem. 
For the sake of smooth readability, proofs of these lemmas are placed in the appendix. 

\begin{lemma}%
\label{lem:der_phi}
Let $\phi$ be a Hurwitz polynomial. Then $\dot{\phi}:=\frac{d}{ds}\phi$ is also a Hurwitz polynomial.
\end{lemma}

\begin{lemma}
\label{lem:gcr_pcr_formulas}
Consider a polynomial $\phi(s)=\prod_{i=1}^{n}(s-z_i)$, where {\ti no member} of $\mathcal{Z}:=\{z_i{\ti~|~}i=1,\cdots,n\}$ is on the imaginary axis. Further, we assume 
$\bar{z}$ belongs to $\mathcal{Z}$ if $z$ does; i.e., $\phi$ is a real polynomial. Then
\begin{align}
&\frac{d}{d\omega}\theta_{\phi}(\omega) = \sum_{i=1}^n \frac{-\mathrm{Re}(z_i)(|z_i|^2+\omega^2)}{(|z_i|^2-\omega^2)^2+4\omega^2\mathrm{Re}(z_i)^2}
\label{poly_pcr}\\
&\frac{d}{d\omega}\ln|\phi(j\omega)| = \sum_{i=1}^n \frac{\omega(\omega^2+\mathrm{Re}(z_i)^2-\mathrm{Im}(z_i)^2)}{(|z_i|^2-\omega^2)^2+4\omega^2\mathrm{Re}(z_i)^2}\label{poly_loggcr}
\end{align}
where $\theta_{\phi}(\omega):=\angle\phi(j\omega)$. 
\end{lemma}

\begin{lemma}
\label{lem:convexity_gen}
Let $\phi$ be a polynomial that satisfies $\phi(0)>0$ and $\frac{d}{d\omega}\mathrm{Re}(\phi(j\omega))\not=0$
for $\omega\in(0,\omega_{\pi})$.  
Then $\mathcal{R}_{\pi}(\phi)$ is convex if 
$\frac{d}{d\omega}\theta_{\dot{\phi}}(\omega) > 0$ and $\frac{d}{d\omega}\mathrm{Re}(\phi(j\omega)) < 0$ for $\omega\in(0,\omega_{\pi})$, where $\theta_{\dot{\phi}}(\omega):= \angle\dot{\phi}(j\omega)$, and $\dot{\phi}(s)
:=\frac{d}{ds}\phi(s)$. 
\end{lemma}

\begin{lemma}
\label{lem:tech_cond4}
Let $\phi:=\phi_1\phi_2$. Suppose $\phi_i$, $i=1,2$, are gain- and phase-monotone increasing and satisfy 
\begin{align}
\label{ineq:gen_pcr_generic}
|\phi_i|(|\phi_i|)' \le \omega |\phi_i'|^2, \ \ \forall\omega\in[0,\infty).
\end{align}
Then the same properties hold for $\phi$. 
\end{lemma}

\begin{lemma}
\label{lem:tech_cond5}
Let $\phi:=\phi_1\phi_2$. Suppose $\phi_i$, $i=1,2$, are gain-monotone increasing and satisfy 
\begin{align}
\label{ineq:2nd_deriv}
F_i''(\Omega) \ge 0, \ \ \forall\omega\in[0,\infty), \ \ i=1,2;
\end{align}
where $F_i(\Omega):=|\phi_i(j\omega)|^2$. Then the same properties hold for $\phi$. 
\end{lemma}

\subsubsection*{Proof of Theorem~\ref{thm:main_class_phi}}
By~\eqref{poly_pcr}, the phase change rate of $\phi(j\omega)$ is positive for all $\omega$ if $\mathrm{Re}(z_i)<0$. Thus,
Hurwitz polynomials are phase-monotone increasing. 

To show $\phi$ is convex, first note that by Lemma~\ref{lem:der_phi}, $\dot{\phi}$ is a Hurwitz polynomial since $\phi$ is. 
Thus, by~\eqref{poly_pcr} again we have $\frac{d}{d\omega}\theta_{\dot{\phi}}(\omega) > 0$ for $\omega\in (0,\omega_{\pi})$ 
(in fact, for all $\omega$). Now let 
\begin{align*}
&\phi(j\omega):=\prod_{k=1}^{m} (j\omega-q_i)\prod_{\ell=1}^{n} \left((j\omega)^2-2\mathrm{Re}(z_{\ell})(j\omega)+|z_{\ell}|^2\right)\\
&= \left(\prod_{k=1}^m r^{r}_k(\omega)\prod_{\ell=1}^n r^c_{\ell}(\omega)\right)\mathrm{exp}\left( j\sum_{k=1}^m\theta^r_k(\omega)+ j \sum_{\ell=1}^n\theta^c_{\ell}(\omega)\right)\\
& =: r_{\phi}(\omega)\mathrm{exp}(j\theta_{\phi}(\omega)),
\end{align*}
where
\begin{align*}
&r^r_{k}:=|j\omega-q_{k}|; \ \ \theta^r_{k}(\omega):=\angle (j\omega-q_{k}), \ \ k\in\{1,\cdots,m\}.\\
&r^c_{\ell}:=\left||z_{\ell}|^2-\omega^2-2\mathrm{Re}(z_{\ell})(j\omega)\right|; \\
&\theta^c_{\ell}(\omega):=\angle\left(|z_{\ell}|^2-\omega^2-2\mathrm{Re}(z_{\ell})(j\omega)\right), \ \ \ell\in\{1,\cdots,n\}.
\end{align*}
In the followings, we use $\square'$ to denote the derivative of $\square$ w.r.t. $\omega$, and drop the $\omega$-dependency for notational clarity. 

Note that $\mathrm{Re}(\phi)' = r_{\phi}'\cos\theta_{\phi} - r_{\phi}\theta_{\phi}'\sin\theta_{\phi}$. 
Thus, $\mathrm{Re}(\phi)'\le 0$ if and only if $r_{\phi}'\cos\theta_{\phi}\le r_{\phi}\theta_{\phi}'\sin\theta_{\phi} $,
which in turn is equivalent to $(r_{\phi}'/r_{\phi})\cos\theta_{\phi} \le \theta_{\phi}'\sin\theta_{\phi}$. Moreover,
notice that $r_{\phi}'/r_{\phi} = \sum_{k=1}^m (r^r_k)' / r^r_k + \sum_{\ell=1}^m (r^c_\ell)' / r^c_\ell$. We will 
now proceed to show that, for all $\omega\in [0,\omega_{\pi}]$, 
\begin{align}
\left(\sum_{k=1}^m \frac{(r^r_k)'}{r^r_k} + \sum_{\ell=1}^n \frac{(r^c_\ell)'}{r^c_\ell} \right)\cos\theta_{\phi} \le \theta_{\phi}'\sin\theta_{\phi}.
\label{eq:95}
\end{align}

First notice that $\mathrm{Re}(j\omega-z_i)'\equiv 0$ and $\mathrm{Re}(|z_{\ell}|^2-\omega^2-2\mathrm{Re}(z_{\ell})(j\omega))' = -2\omega <0$
for all $\omega>0$. Therefore, we have for all $\omega>0$, 
\begin{align*}
&\left( (r^r_k)'/r^r_k \right)\cos\theta^r_k = (\theta^r_k)' \sin\theta^r_k, \ \ k\in\{1,\cdots,m\}{\ti,}  \\ 
&\left( (r^c_\ell)'/r^c_\ell \right)\cos\theta^c_\ell < (\theta^c_\ell)' \sin\theta^c_\ell, \ \ \ell\in\{1,\cdots,n\}{\ti.}
\end{align*}

Now for $\omega\in[0,\omega_{\pi}]$, if $\theta_{\phi}\in [\pi/2,\pi)$, we have $\cos\theta_{\phi}\le 0 < \sin\theta_{\phi}$, and thus
\begin{align}
&\left( (r^r_k)'/r^r_k \right)\cos\theta_{\phi} < (\theta^r_k)' \sin\theta_{\phi}, \ \ k\in\{1,\cdots,m\} \label{eq:96}  \\ 
&\left( (r^c_\ell)'/r^c_\ell \right)\cos\theta_{\phi} < (\theta^c_\ell)' \sin\theta_{\phi} \ \ \ell\in\{1,\cdots,n\}, \label{eq:97} 
\end{align}
which in turn leads to~\eqref{eq:95} when both sides of~\eqref{eq:96} and~\eqref{eq:97} 
are summed over $k$ and $\ell$, respectively, and then added together.  

If $\theta_{\phi}\in (0, \pi/2)$, use the fact that $\cos(\cdot)$ is a decreasing function and $\sin(\cdot)$ increasing, and 
$\theta_{\phi} > \theta^r_k$, $\theta_{\phi} > \theta^c_\ell$ for each $k$ and $\ell$, to deduce~\eqref{eq:96} and~\eqref{eq:97}. 
Then~\eqref{eq:95} follows per the same arguments. 

That $\phi$ is gain-monotone increasing if $\phi\in\mathcal{H}_p$, is a direct consequence of~\eqref{poly_loggcr}. 
Finally, by Lemmas~\ref{lem:tech_cond4} and~\ref{lem:tech_cond5}, we can conclude that
$\phi\in\mathcal{H}_p$ implies that $\phi$ satisfies~\eqref{phi_cond4} and~\eqref{phi_cond5} if we may verify {\ti that} the following
first and second order factors satisfy the two inequalities{\ti:}
\begin{align}
&\varphi_1(s)=s+a ,\\
&\varphi_2(s)=s^2-2\mathrm{Re}(z)s+|z|^2,
\end{align}
where $a>0$ and $z$ satisfies $\mathrm{Re}(z)<0$ and $|\mathrm{Im}(z)|\le|\mathrm{Re}(z)|$. Note that
$F_1(\Omega):=|\varphi_1(j\omega)|^2 = \omega^2+a^2 = \Omega + a^2$, and $F_2(\Omega):=|\varphi_2(j\omega)|^2 
= (|z|^2-\omega^2)^2 + 4\mathrm{Re}(z)^2\omega^2 = (|z|^2-\Omega)^2 + 4\mathrm{Re}(z)^2\Omega$. One can readily
verify that $F_1$ and $F_2$ satisfy~\eqref{phi_cond5}. For inequality~\eqref{phi_cond4}, note the following equalities hold:
$|\varphi_1|(|\varphi_1|)' = \omega$, $\omega |\varphi_1'|^2 = \omega$,
\begin{align*}
&|\varphi_2|(|\varphi_2|)' = 2\omega(\omega^2+\mathrm{Re}(z)^2-\mathrm{Im}(z)^2), \\ 
&\omega |\varphi_2'|^2 = 2\omega(2\omega^2+2\mathrm{Re}(z)^2).
\end{align*}
Clearly, $\varphi_1$ and $\varphi_2$ both satisfy inequality~\eqref{phi_cond4}. This concludes the proof.

\section{Rank-One and Rank-Two Networks}
\label{sec:LowRank}

This section investigates the robust instability analysis for the multi-agent system 
$\Sigma(\Delta, H, A)$ depicted in Section~\ref{sec:PF}, with an additional structural assumption 
that matrix $A$ is of rank one or rank two.  
That is, $A$ can be factorized as $L\Lambda_kR$ with $L\in\mathbb{C}^{n\times k}$, 
$R\in\mathbb{C}^{k\times n}$, 
$\ti\Lambda_k=\diag(\lambda_1,\cdots,\lambda_k)$, and $k=1$ or $2$.
Note that in the case $k=1$, $\lambda_1$ is real and $L\in\mathbb{R}^{n\times 1}$, $R\in\mathbb{R}^{1\times n}$. 
In the case $k=2$, {\ti we assume} $\lambda_1$ and $\lambda_2$ are complex conjugates, i.e., $\lambda_2 = \bar{\lambda}_1$. 
We are {\ti particularly} interested in these low-rank networks due to the fact that the 
RIR problem with such networks degenerates to a SISO problem, which will be elaborated in the 
next two subsections. 

By {\ti Lemma~1} of~\cite{NetworkRIR_short}, $h, \ w_m\in\RHinf$ is necessary 
for the network to be stable {\ti due to the rank deficiency of $A$}. Further, by 
{\ti Lemma~2 of \cite{NetworkRIR_short}}, internal stability of
$\Sigma(\Delta, H, A)$ is equivalent to that of the feedback system $[\![G_k, R\Delta L]\!]$, where
\begin{align}
\label{eq:reduced_sys}
G_k = \mathrm{diag}(g_1,\cdots,g_k), \ \ {\ti g_i} =\frac{\lambda_i w_mh}{1-\lambda_ih}, \ 
\ti \ i\in\II_k. %
\end{align}
Thus, the network RIR problem reduces to 
\begin{align}
\inf_{\Delta}~\|\Delta \|_{\infty}, \ \  \mbox{subj. to}~~ \Delta\in\mathbf{\Delta}_d, \ \ R\Delta L\in \mathbb{S}_f(G_k){\ti.} 
\label{eq:siso_rir}
\end{align} 
For the development of our main results in this section, we also consider the following problems
\begin{align}
&\inf_{\hat{\Delta}_k}  ~\|\hat{\Delta}_k\|_{\infty}, \ \ \mbox{subj. to}~~ \hat{\Delta}_k\in\mathbb{S}_f(G_k) 
\label{2steps:1}\\
&\inf_{\Delta} ~\|\Delta\|_{\infty}, \ \  \mbox{subj. to}~~ \Delta\in\mathbf{\Delta}_d, \ \ R\Delta L = \hat{\Delta}^*,
\label{2steps:2}
\end{align}
where $\hat{\Delta}^*$ is an optimizer of~\eqref{2steps:1}.

\subsection{The Rank-One Case}

In the case of $k=1$, $G_k := g_1$ is a scalar transfer function. 
We further assume that $g_1\in\RLinf$ and is unstable. 
For rank-one networks, the following technical lemma is required for
deriving the main results of this section.
\begin{lemma}
\label{lem:equaldist}
Given $\delta\in\RHinf$ {\ti and} $c_1,\cdots,c_n\in\mathbb{R}$, 
\begin{align*}
\inf_{\begin{subarray}{c}
\delta_i\in\RHinf, \ti i\in\II_n \\ %
\end{subarray} 
}\max \ \{\|\delta_i\|_\infty, {\ti i\in\II_n\}} %
\ti \mbox{ subj. to } \sum_{i=1}^n c_i\delta_i = \delta  
\end{align*}
is equal to $\gamma^*:=\|\delta\|_{\infty}/(|c_1|+\cdots+|c_n|)$. 
\end{lemma}
\begin{proof}
First we show that $\gamma^*$ is a lower bound on the infimum in question. 
Suppose not. This means
there exist $\tilde{\delta}_1, \cdots, \tilde{\delta}_n$ such that 
$c_1\tilde{\delta}_1+\cdots+c_n \tilde{\delta}_n = \delta$ and 
\begin{align*}
\max\{\|\tilde{\delta}_i\|_\infty, i=1\cdots, n\} < \gamma^*
\end{align*}
This in turn implies 
$|\tilde{\delta}_i(j\omega)| < \gamma^*$, $\forall \omega$, $\forall i=1,\cdots, n$. Suppose 
$\|\delta\|_{\infty} = |\delta(j\omega_p)|$. We have 
\begin{align*}
\|\delta\|_{\infty} &= |\delta(j\omega_p)|=|c_1\tilde{\delta}_1(j\omega_p)
+\cdots+c_n \tilde{\delta}_n(j\omega_p)|\\
&\le |c_1||\tilde{\delta}_1(j\omega_p)|+\cdots+|c_n| |\tilde{\delta}_n(j\omega_p)| \\ 
&< (|c_1|+\cdots+|c_n|)\gamma^* = \|\delta\|_{\infty},
\end{align*}
which is a contradiction. Second, note that $\gamma^*$ is attained by the following selection of $\delta_i$'s: 
\begin{align*}
\delta_i^* = \mathrm{sign}(c_i) \cdot \frac{\delta}{|c_1|+\cdots+|c_n|}, \ \ i=1,\cdots, n,
\end{align*}
where $\mathrm{sign}(c)=1$ if $c> 0$, $\mathrm{sign}(c)=-1$ if $c < 0$, and 
$\mathrm{sign}(0)=1$ or $-1$.
Clearly, such $\delta_i^*$, $i=1,\cdots,n$, satisfy the equality constraint, and 
$\|\delta_1\|_{\infty}=\|\delta_2\|_{\infty}\cdots=\|\delta_n\|_{\infty}=\gamma^*$.
\end{proof}

\begin{remark}
We note that all $\delta_i^*$ are the same if all $c_i$ have the same sign; i.e., 
they are all non-negative, or all non-positive. In either case, the sign of the zero
coefficients shall be chosen to be aligned with that of the nonzero coefficients. 
\end{remark}

The following lemma is instrumental for the main result of this section.
\begin{lemma}
\label{lem:two-step} 
For $k=1$, the infimum of problem~\eqref{eq:siso_rir} is equal to that of 
problem~\eqref{2steps:2}, where the argument of infimum of problem~\eqref{2steps:1} is 
applied to constrain $\Delta$. 
\end{lemma}

\begin{proof}
First, note that the infimum of problem~\eqref{2steps:2} is larger than or equal to the infimum of 
problem~\eqref{eq:siso_rir}. To see this, notice that the argument of infimum of problem~\eqref{2steps:2}
is a feasible solution to problem~\eqref{eq:siso_rir}. 

To see the opposite inequality, let the argument of infimum of problem~\eqref{eq:siso_rir} be 
$\Delta^*:=\mathrm{diag}\left(\delta_1^*,\cdots,\delta_n^*\right)$ with 
$\|\Delta^*\|_{\infty}=\gamma^*$, and let
$c_i:=r_i\ell_i$, where $r_i$ and $\ell_i$ are the $i^{\rm th}$ entries of $R$ and $L$, respectively. Furthermore, let  
the argument of infimum of~\eqref{2steps:1} be $\hat{\Delta}_k^*$. 
Since $c_1\delta_1^*+\cdots+c_n\delta_n^*$ is {\ti a} feasible solution
of problem~\eqref{2steps:1}, we have 
\begin{align*}
&\|\hat{\Delta}_k^*\|_{\infty}\le \|c_1\delta_1^*+\cdots+c_n\delta_n^*\|_{\infty} \\
&\le|c_1|\|\delta_1^*\|_{\infty}+\cdots+|c_1|\|\delta_n^*\|_{\infty}
\le (|c_1|+\cdots+|c_n|)\gamma^*.
\end{align*}
By Lemma~\ref{lem:equaldist}, the infimum of problem~\eqref{2steps:2} is equal to 
$\|\hat{\Delta}_k^*\|_{\infty}/(|c_1|+\cdots+|c_n|)$, and hence it is less than or equal to $\gamma^*$. 
This concludes the proof.
\end{proof}

The main result of this section is stated below. 
\begin{prop}
\label{prop:rir_rank1}
Consider the multi-agent system $\Sigma(\Delta,H,A)$ depicted in Section~\ref{sec:PF} with a rank-one 
matrix $A$. Suppose $A$ is diagonalizable and equal to $L\cdot\lambda_1\cdot R$, 
where $L,\ R^T \in\mathbb{R}^n$ and $RL=1$. Furthermore, suppose  
$g_1:=\frac{\lambda_1 w_m h}{1-\lambda_1 h}$ satisfies one of the following conditions:
\begin{itemize}
\item $g_1\in\G_1^0$ and $\theta_{g_1}'(0)> 0$; 
\item $g_1\in\G_2^{\pm\omega_p}$ and $\theta_{g_1}'(\omega_p)>|\sin(\theta_{g_1}(\omega_p))/\omega_p|$.
\end{itemize}
Then the nominal network $\mathcal{F}_\ell(H,A)$ has the exact dynamic RIR 
$\rho_*(\LFT(H,A))=\|g_1\|_{\infty}^{-1}/(|c_1|+\cdots+|c_n|)$,
where $c_i=r_i\ell_i$, $i=1,\cdots, n$, {\ti and} $r_i$ and $\ell_i$ are {\ti the} $i^{\rm th}$ {\ti entries} of $R$ and $L$, respectively. 
\end{prop}
\begin{proof}
Note that by definition, $\rho_*(\mathcal{F}_\ell(H,A))$ is the infimum of~\eqref{eq:siso_rir} (with $G_k = g_1$), 
which by Lemma~\ref{lem:two-step} is equal to the infimum of~\eqref{2steps:2} with
$\Delta$ constrained by the argument of infimum of~\eqref{2steps:1} (with $G_k = g_1$). The infimum of~\eqref{2steps:1}, 
on the other hand, is by definition $\rho_*(g_1)$. 
Therefore, by Lemma~\ref{lem:equaldist}, the infimum of~\eqref{2steps:2} is equal to $\rho_*(g_1)/(|c_1|+\cdots+|c_n|)$. 
Now by Theorem~\ref{thm:MS}, the conditions $g_1$ is assumed to satisfy implies that $g_1$ can be marginally stabilized
by a stable controller with norm equal to $1/\|g_1\|_{\infty}$. Therefore, $\rho_*(g_1) = 1/\|g_1\|_{\infty}$ and
we conclude our proof. 
\end{proof}

\begin{remark}
Proposition~\ref{prop:rir_rank1} assumes that $A$ is diagonalizable. In the case where
$A$ is rank-one but not diagonalizable, $A$ can be expressed as $L R$ with $RL=0$. 
Then, the RIR problem of finding $\rho_*(\mathcal{F}_\ell(H,A))$ is equivalent to finding the smallest 
stable $\Delta$ such that the feedback system $[\![(R\Delta L), \lambda_1 w_m h]\!]$ is stable. 
To see this, note that the characteristic equation of $[\![\Delta, \mathcal{F}_\ell(H,A)]\!]$ is equal to
$\mathrm{det}(1-\lambda_1 w_m h (R\Delta L) )$, the same as that of $[\![(R\Delta L), \lambda_1 w_m h]\!]$. 
Thus, the smallest $\Delta$ is zero under the necessary condition that $h$ and $w_m$ belong to $\RHinf$. 
{\ti That is, the nominal network with $\Delta=0$ is necessarily stable in this case
and the RIR analysis is trivial.}
\end{remark}

\begin{remark}
\label{rmk:rank-1-pcr}
Following Proposition~\ref{prop:gen_pcr} and Theorem~\ref{thm:main_class_phi}, $g_1$ would satisfy the respective PCR conditions
if $\phi:=1/h$ belongs to the class $\mathcal{H}_p$ defined in Theorem~\ref{thm:main_class_phi}. 
\end{remark}

\subsection{Rank-Two $A$ with Complex Conjugate Eigenvalues}

Now consider the case where the interconnection matrix $A$ is of rank two with a pair of complex eigenvalues. 
As such, $A=L\Lambda_2 R$, where $\Lambda_2:=\mathrm{diag}(\lambda, \bar{\lambda})$ and $RL = I_2$. 
This leads to $G_k=\mathrm{diag}(g_1, g_2)$, with
\begin{align}
\label{eq:rank2-gs}
g_1(s):=\frac{\lambda w_m h(s)}{1-\lambda h(s)}, \ 
g_2(s)= \frac{\bar{\lambda} w_m h(s)}{1-\bar{\lambda} h(s)} =  \overline{g_1(\bar{s})} 
\end{align}

The following technical lemma is crucial for developing the main result of this subsection. 
\begin{lemma}
\label{lem:conjg}
The following hold concerning $g_1$ and $g_2$ introduced above. 
\begin{enumerate}
\item[(i)] $g_1\in\mathcal{RL}_\infty$ if and only if $g_2\in\mathcal{RL}_\infty$; furthermore,
$\|g_1\|_{\infty} = \|g_2\|_{\infty}$;
\item[(ii)] $g_1$ is unstable if and only $g_2$ is. 
\item[(iii)] {\ti Given} any real rational $\delta$ and a complex number $s$,  $1-g_1(s)\delta(s)=0$ if and only if $1-g_2(\bar{s})\delta(\bar{s})=0$.  
\item[(iv)] Let $\delta$ be real rational, then $\delta$ marginally stabilizes $g_1$ if and only if it marginally stabilizes $g_2$.
\end{enumerate}
\end{lemma}
\begin{proof}
To see (i), note that $\overline{g_1(j\omega)} = g_2(-j\omega)$ for all $\omega$, and therefore
$|g_1(j\omega)| = |g_2(-j\omega)|$. Hence if $g_1$ is bounded over the imaginary axis, so is $g_2$, 
and the supremums of $|g_1(j\omega)|$ and $|g_2(j\omega)|$ over $\omega$ would be the same. To see (ii),
note that the complex conjugates of the poles of $g_1$ are the poles of $g_2$ (and vice versa). Since 
a pair {\ti of} complex {\ti conjugates} have the same real part, the aforementioned implies that $g_1$ has an unstable
{\ti pole} if and only if $g_2$ has one. To see (iii), notice that since $\delta$ is real rational, we have
$\overline{1-g_1(s)\delta(s)} = 1 - \overline{g_1(s)}\delta(\overline{s}) 
= 1 - g_2(\overline{s})\delta(\overline{s})$. 

To prove statement (iv), we note that if all the roots of $1-g_1(s)\delta(s)$ are in the {\ti CLHP}, 
so are the roots of $1-g_2(s)\delta(s)$, and vice versa. To see this, suppose one of the
roots of $1-g_2(s)\delta(s)$ (say, $\zeta$) has a positive real part. Since $\overline{\zeta}$ would be 
a root of $1-g_1(s)\delta(s)$, we reach a contradiction. Finally, statement (iii) also implies that
the multiplicities of the roots of $1-g_1(s)\delta(s)$ are identical to those of $1-g_2(s)\delta(s)$.
\end{proof}

The main result of this subsection is stated next.
\begin{prop}
\label{prop:rir_rank2}
Consider the multi-agent system $\Sigma(\Delta,H,A)$ depicted in Section~\ref{sec:PF} with a rank-two 
matrix $A$. Suppose $A$ is %
equal to $L \mathrm{diag}(\lambda,\overline{\lambda}) R$ with $RL=I_2$. Let $G_2:=\mathrm{diag}(g_1, g_2)$
with $g_1, \ g_2$ as defined in~\eqref{eq:rank2-gs}. 
Suppose 
\begin{enumerate}
\item[(i)]  $A$ is diagonally normalizable\footnote{$A$ is said to be diagonally normalizable if there exists a complex diagonal matrix
$D$ such that $D^{-1}AD$ is normal. Note that any $2\times 2$ rank-two matrix is diagonally normalizable if and only if its diagonal entries
are the same or its off-diagonal entries have the same sign.};
\item[(ii)] $g_1\in\G_1^{\omega_p}$ with $\omega_p\not=0$, and satisfies $\theta_{g_1}'(\omega_p)>|\sin(\theta_{g_1}(\omega_p)/\omega_p|$. 
\end{enumerate}
Then the nominal network $\mathcal{F}_\ell(H,A)$ has the exact RIR $\rho_*(\mathcal{F}_\ell(H,A)) = 
1/\|g_1\|_{\infty}$. In this case, $\rho_*(\mathcal{F}_\ell(H,A)) = \rho_h(\mathcal{F}_\ell(H,A))$,
i.e., the minimum{\ti-norm} stabilizing diagonal perturbation is homogeneous.
\end{prop}

\begin{proof}
Note that by Theorem~\ref{thm:MS}, assumption (ii) implies that $g_1$
can be marginally stabilized by a $\delta^*\in\RHinf$ with $\|\delta^*\|_{\infty} = 
1/\|g_1\|_{\infty}$. By Lemma~\ref{lem:conjg}, $g_2$ can also be marginally stabilized by
$\delta^*$, and therefore $G_2$ {\ti can} be marginally stabilized by $\delta^*I_2$. 
Furthermore, since $\|g_2\|_{\infty} =\|g_1\|_{\infty}$, we have $\|G_2\|_{\infty}
=\|g_1\|_{\infty}$ as well. Note that by the small-gain theorem, 
the infimum of problem~\eqref{2steps:1} (with $G_2$ replacing $G_k$) cannot be smaller than $\|G_2\|_{\infty}^{-1}$. Therefore,
the infimum of problem~\eqref{2steps:1} is equal to $\|g_1\|_{\infty}^{-1}$ with the argument of 
infimum being $\delta^*I_2$ (i.e., $\hat{\Delta}_k^* = \delta^*I_2$). 

As the argument of the infimum of~\eqref{2steps:1} is diagonal and homogeneous, with 
assumption (i) we may apply Proposition~4 of~\cite{NetworkRIR_short} to conclude that an optimal solution 
to~\eqref{eq:siso_rir} is $\delta^*I_n$. Since, by definition, the infimum of~\eqref{eq:siso_rir} is 
$\rho_*(\LFT(H,A))$, we conclude that $\rho_*(\LFT(H,A))=\|\delta^*\|_{\infty}=1/\|g_1\|_{\infty}$,
and clearly $\rho_*(\LFT(H,A)) = \rho_h(\LFT(H,A))$ since the argument of infimum is homogeneous.
\end{proof}

\begin{remark}
\label{rmk:rank-2-pcr}
Similar to the rank-one scenario, the PCR condition stated in Proposition~\ref{prop:rir_rank2} would be satisfied 
if $\phi:=1/h$ belongs to the class $\mathcal{H}_p$ defined in Theorem~\ref{thm:main_class_phi}. 
\end{remark}

\section{A Biological Application}
\label{sec:BioAppli}

In this section, we demonstrate the robust instability analysis using the model of oscillatory biomolecular reactions.
Specifically, we consider a genetic regulatory network (GRN) with a cyclic network structure depicted in Fig. \ref{fig:cyclicGRN}. 
In this system, each protein $P_i$ is produced through the processes of transcription and translation of gene $i$ and represses the production of another, forming the cyclic network. 
The dynamics of the transcription and translation process (gene expression) are modeled by 
\begin{equation}
  \begin{aligned}
  &\dot{r}_i(t) = -a r_i(t) + \beta_i u_i(t) + \beta_i \hat{w}_m (\hat{\delta}_i u_i(t))\\
  &\dot{p}_i(t) = -b p_i(t) + c_i r_i(t),
  \end{aligned}
  \label{eq:GRN-model}
\end{equation}
where $r_i(t)$ and $p_i(t)$ represent the concentrations of the messenger RNA (mRNA) and the protein associated with gene $i$, respectively. 
The input $u_i(t)$ denotes the repressive regulation by the protein $P_{i-1}$ defined by 
\begin{align}
  u_i(t) = \frac{K^\nu}{K^\nu + p_{i-1}^\nu} (=: \psi(p_{i-1})), 
  \label{eq:Hill-def}
\end{align}
where $K$ and $\nu$ is the dissociation constant and the Hill coefficient, respectively, and the subindex $i$ is taken modulo $N$, i.e., $p_0 = p_N$. 
Unlike the original model in the literature \cite{Niederholtmeyer2015}, \eqref{eq:GRN-model} considers multiplicative uncertainty represented by $\hat{\delta}_i$ with the weight function $\hat{w}_i$, 
where $\hat{\delta}_i$ and $\hat{w}_i$ represent the linear operators with the input-output mapping specified by stable transfer functions $\delta_i(s)$ and $w_m(s)$, respectively. 

\begin{figure}[tb]
  \centering
  \includegraphics[clip, width=0.7\columnwidth]{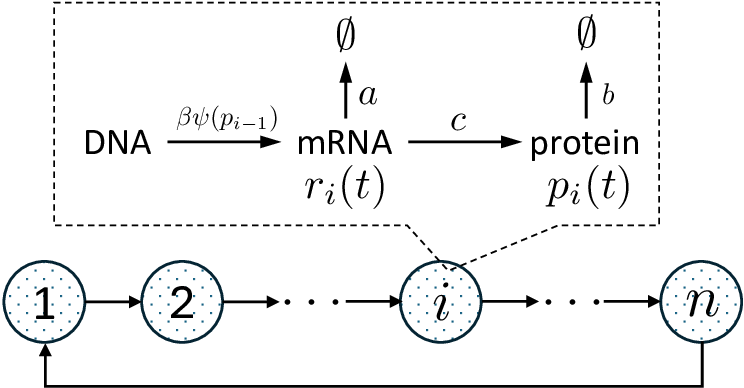}
  \caption{Genetic regulatory network with cyclic network structure considered in the example.}
  \label{fig:cyclicGRN}
\end{figure}

When the number of protein species $N$ is odd, the GRN with the cyclic network has a unique equilibrium point $\bm{x}^* := [r_1^*, p_1^*, \ldots, r_N^*, p_N^*]\in \mathbb{R}^{2N}$ \cite{Hori2011}. 
The local instability of the equilibrium $\bm{x}^*$ leads to sustained oscillations of the mRNA and protein concentrations \cite{Hori2011}, 
which was experimentally demonstrated for $N=3$ and $N=5$ using synthetic GRNs \cite{Niederholtmeyer2015,Elowitz2000}.
Thus, to guarantee oscillations amid uncertainties, it is crucial to quantify the RIR.

For this purpose, we consider the linearized model of the cyclic GRN around the equilibrium point $\bm{x}^*$. 
We define $h(s)$ as the transfer function from $u_i(t)$ to $p_i(t)$ in \eqref{eq:GRN-model}
and $A \in \mathbb{R}^{N \times N}$ as the matrix of the linearized gain of $\psi(\cdot)$ in \eqref{eq:Hill-def}. 
{\red 
Specifically, $h(s)$ is represented by \eqref{eq:h-grn} 
and $A$ is a cyclic matrix with entries defined by:
\begin{align}
A_{ij} = 
\begin{cases} 
\psi'(p_N^*), & (i,j) = (1,N) \\ 
\psi'(p_j^*), & j = i-1 \text{ for } i = 2, \dots, N \\ 
0, & \text{otherwise.} 
\end{cases}\label{eq: matrixA}
\end{align}
}
Our goal is then to find the RIR for the cyclic network system $\Sigma(\Delta, H, A)$, i.e., a minimum-norm $\delta(s)$ that stabilizes the linearized system. 

In what follows, we consider the case of $N=7$ and set the parameters in model \eqref{eq:GRN-model} as shown in \eqref{eq:param-grn} and $\nu=2$. %
Then, the unique equilibrium point is given by $r_i^* = 0.2949$ and $p_i^* = 19.1424\ (i=1,2,\ldots,7)$, which defines the matrix $A$ in eq. \eqref{eq: matrixA}  
with $\psi'(p_{i-1}^*) = -0.006247$.
As shown in Fig. \ref{fig:four_images}, the nominal system is unstable with this set of parameters, and the concentrations $r_i^*$ and $p_i^*$ exhibit sustained oscillations. 

In general, the matrix $A$ depends on the equilibirum point $\bm{x}^*$, which is affected by the static gain of the uncertainties $w_m(0)\delta_i(0)\ (i=1,2,\ldots,7)$. 
This introduces additional complexity to the RIR analysis (see \cite{hara2022instability} for the case of a single uncertainty). 
However, to focus on the demonstration of the RIR analysis for the linearized system, we assume $w_m(0)=0$ in this example. 
Specifically, we choose a high-pass weight function of the form $w_m(s) = s/(s+\xi)$,
where $\xi$ is set 0.001 to ensure $|w_m(j\omega_p)| \simeq 1$. 

According to Theorem \ref{thm:homo_equiv_hetero}, the RIR is obtained as $1/\|g_1\|_\infty = 0.3907$ with 
the stabilizing perturbation
\begin{align}
  & \delta_i(s) = 0.3907 (1+\epsilon) \frac{s-1.814}{s+1.814} (=: \delta(s)). 
  \label{eq:delta-example}
\end{align}
and $\epsilon = 0$. 
To verify this result, we first compute the poles of the closed-loop system with uncertainty. 
Specifically, the characteristic equations $1-\delta(s)g_1(s)=0$ and $1-\delta(s)g_7(s)=0$ each has a single root at $0+0.31j$ and $0-0.31j$, respectively, while all the other roots of 
$1-\delta(s)g_i(s)=0\ (i=1,2,\ldots,7)$ are located in the OLCP. 
It should be noted that $g_i(s)\ (i=1,2,\ldots,7)$ are functions with complex coefficients, and thus, the complex roots of each characteristic equation do not necessarily appear as conjugate pairs.
Next, we slightly increase the gain of the perturbation by setting $\epsilon = 0.05$ and find that the system is stabilized by $\delta(s)$. 
Indeed, the time-course simulations of eq. (\ref{eq:GRN-model}) in Figure \ref{fig: oscillator-sim} show the concentrations converge to the equilibrium point when $\epsilon = 0.05$ while they continue to 
oscillate when $\epsilon = -0.05$.

\begin{figure}[tb]
\centering   
\begin{subfigure}{0.47\columnwidth}
    \centering
    \includegraphics[width=\textwidth,clip]{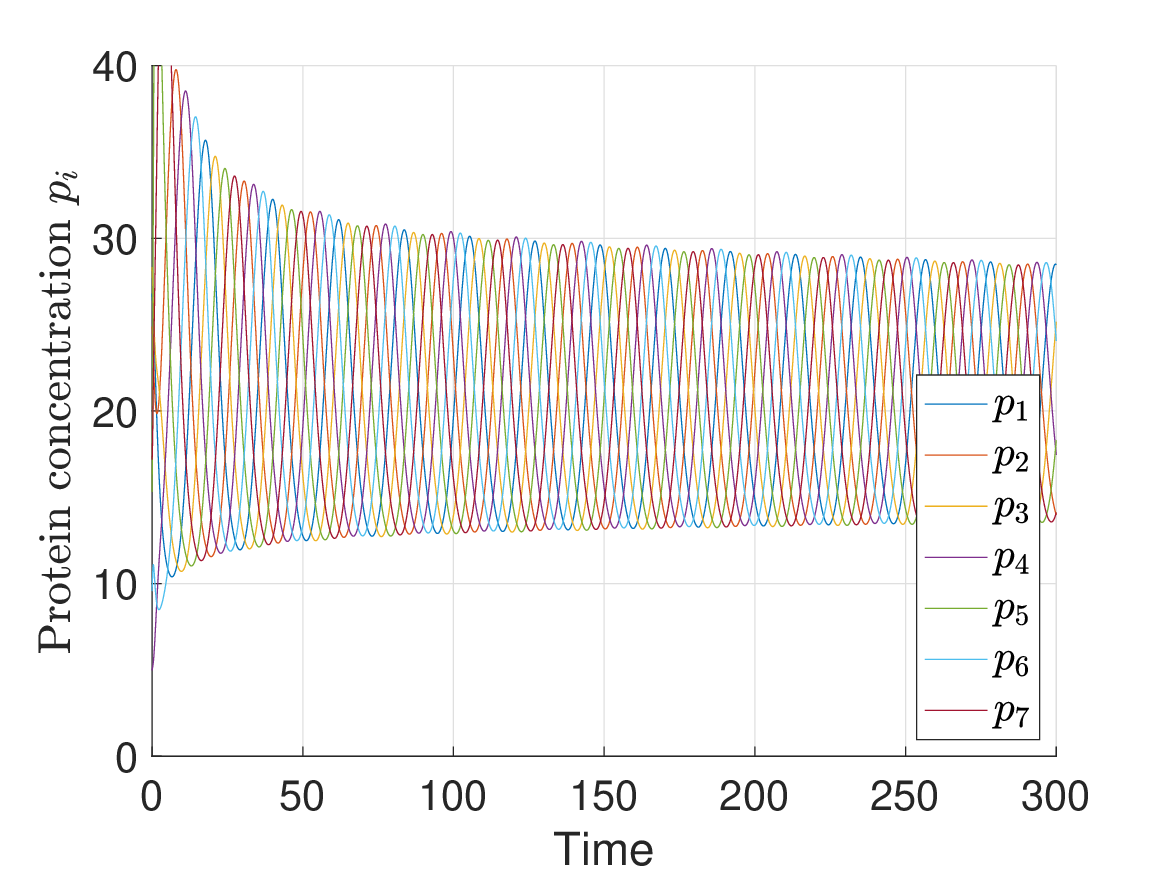}
    \caption{$\epsilon = -0.05$}
\end{subfigure}
\begin{subfigure}{0.47\columnwidth}
    \centering
    \includegraphics[width=\textwidth,clip]{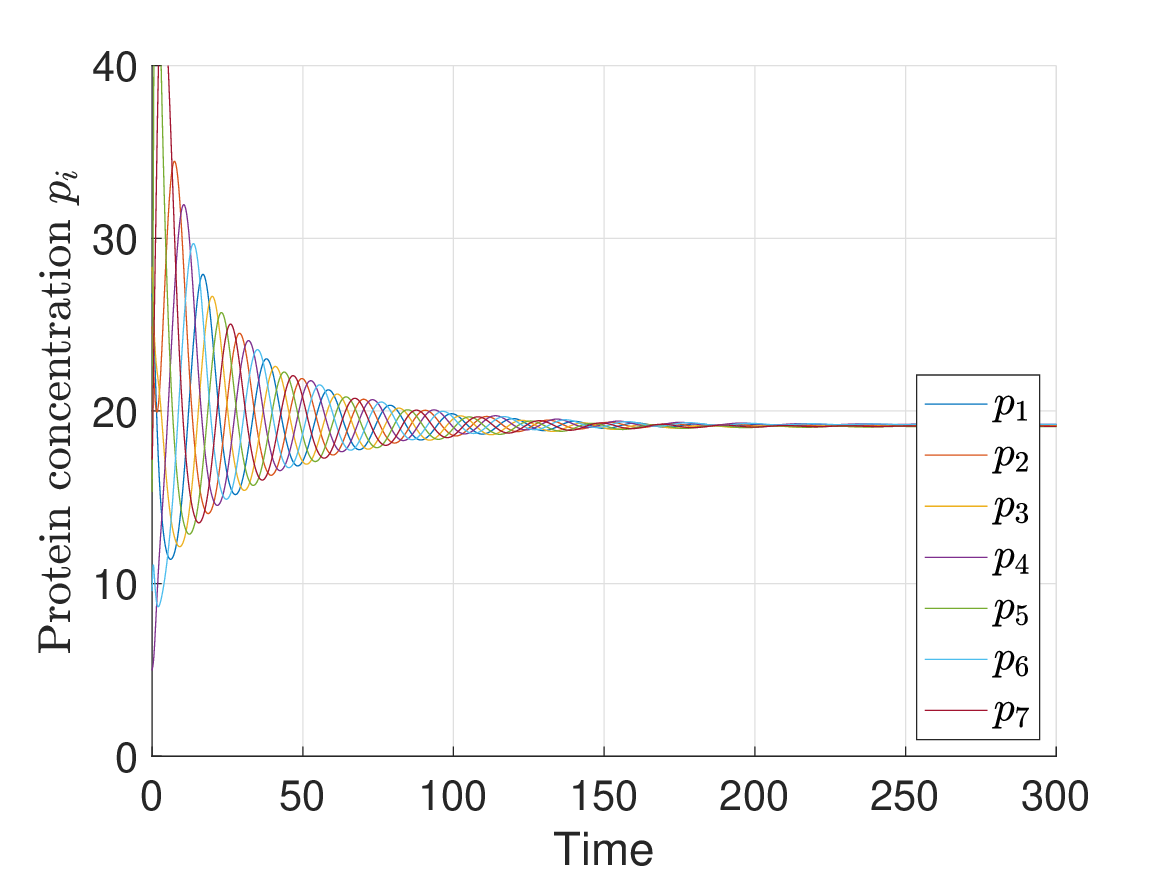}
    \caption{$\epsilon=0.05$}
\end{subfigure}
\caption{Time-course of the proteins $p_i(t)$ for $n=7$} 
\label{fig: oscillator-sim}
\end{figure}

\section{Conclusion}
\label{sec:Concl}
We identified three classes of LTI networked systems of which the robust instability radius may be found by analyzing, respectively, one representative SISO agent. In each case, we provided sufficient conditions under which the robust instability radius is exacted. The results were applied to analyze oscillatory biomolecular reactions of a genetic regulatory network (GRN) with a cyclic network structure.

\section*{References}
\appendix
\subsection{Proof of Proposition~\ref{prop:hetro_eq_homo}} 
\label{subsec:proof_hetro_eq_homo}

The following lemma is instrumental in proving Proposition~\ref{prop:hetro_eq_homo}. 

\begin{lemma} \label{lem:1}
Let $r\in\IR$ be given such that $0<r<1$, and consider
\begin{align}
\label{def:bold_Theta}
\bfThe:=\log(1+\bfdel), \hs
\bfdel:=\{\delta\in\IC:\,|\delta|\leq r\,\}.
\end{align}
The set $\bfThe$ is convex.
\end{lemma}

\medskip
\begin{proof}
Since the boundary of $\bfdel$ is characterized by $re^{j\phi}$ with 
$\phi\in(-\pi,\pi]$, the boundary of $\bfThe$ is characterized by 
\begin{equation}
\label{xy}
x+jy=\log(1+re^{j\phi}), 
\end{equation}
with $\phi\in(-\pi,\pi]$, or equivalently, $x=f(\phi)$, $\tan y=g(\phi)$,
\[
f(\phi):=\frac{1}{2}\log(1+r^2+2r\cos\phi), \ \ 
g(\phi):=\frac{r\sin\phi}{1+r\cos\phi}.
\]
In the above notation, note that $y$ is the angle of the complex number
$(1+r\cos\phi)+jr\sin\phi$. Since $|r|<1$, the real part is 
always positive and therefore $y\in[-\pi/2,\pi/2]$. 
Clearly, the boundary of $\bfThe$ is a closed continuous curve that is 
symmetric about the real axis because if a particular point $x_o+jy_o$ is on 
the boundary for $\phio\in[-\pi,\pi]$ in (\ref{xy}), then $x_o-jy_o$ is on 
the boundary for $-\phio\in[-\pi,\pi]$.
Noting that $y=0$ occurs when $\phi=0$ or $\pm\pi$, the boundary of $\bfThe$ 
crosses the real axis twice at $x=\log(1\pm r)$. As $\phi$ increases from
$0$ to $\pi$, the value of $x$ monotonically decreases from $\log(1+r)$ to 
$\log(1-r)$ while $y$ stays in $[0,\pi/2]$ except at the end points, defining 
a function $h$ that describes the boundary by $y=h(x)$.
Thus, we now see that the set $\bfThe$ is the region enclosed by the
curves $y=h(x)$ and $y=-h(x)$ on $x\in[\log(1-r),\log(1+r)]$.

We prove convexity of the set $\bfThe$ by showing that the curve $y=h(x)$ is
concave. To this end, note that
\begin{align*}
&h'(x) = \frac{g'(\phi)\cos^2y}{f'(\phi)} = -\frac{r+\cos\phi}{\sin\phi}, \\
&h''(x) = \frac{1+r\cos\phi}{\sin^2\phi}\cdot
\frac{(1+r\cos\phi)^2+(r\sin\phi)^2}{-r\sin\phi}.
\end{align*}
Noting that the curve $y=h(x)$ is parametrized for $0<\phi<\pi$, 
the curve is concave due to $h''(x)<0$ on the interval $\log(1-r)<x<\log(1+r)$.
\end{proof}

Now we proceed to prove Proposition~\ref{prop:hetro_eq_homo}. Let 
\begin{align*}
&\mathbb{P}:=\{s\in\mathbb{C}~|~\mbox{$s$ is not a pole of ${\delta}_i$ and ${\delta}_i(s)\not=-1$}, \\
&\hspace{6.5cm} i=1,\cdots,n \},
\end{align*}
and note that since ${\delta}_i\in\RHinf$, $\mathbb{C}\backslash\mathbb{P}$ is a set of finite isolated points.
Furthermore, since ${\delta}_i\in\RHinf$ with $\|{\delta}_i\|_{\infty} \le r<1$, for $i=1,\cdots, n$, we have 
$\Re(1+{\delta}_i(s))\ge 1-r > 0$ for all $s$ in the CRHP and $i=1,\cdots,n$. As such, CRHP is a subset of $\mathbb{P}$.

Given ${\Delta}:=\diag({\delta}_1,\cdots,{\delta}_n)\in\mathbf{\Delta}_d$  
with $\|{\Delta}\|_{\infty} < 1$, let ${\delta}_{\rm av}$ be defined pointwise as
\begin{align}
{\delta}_{\rm av}(s):=\mathrm{exp}\left(\frac{1}{n}\sum_{i=1}^n\log(1+{\delta}_i(s))\right)-1
\label{def:delta_av}
\end{align}
for $s\in\mathbb{P}$, and define ${\delta}_{\rm av}(s)$ to be $-1$ for $s$ such that ${\delta}_i(s) =-1$ for some $i$.  
Clearly, ${\delta}_{\rm av}$ is 
well defined and continuous in the CRHP (in particular, continuous on the imaginary axis) and is analytic in
the ORHP. Also, one can readily verify that ${\delta}_{\rm av}(j\omega)=\overline{{\delta}_{\rm av}(-j\omega)}$
for all $\omega\in\mathbb{R}$. Furthermore, since each ${\delta}_i\in\RHinf$, the limit
\begin{align*}
\lim_{\Re(s)\ge 0, \ |s|\to\infty} {\delta}_i(s)
\end{align*}
exists for all $i=1,\cdots,n$, and so does the limit 
$\lim_{\Re(s)\ge 0, |s|\to\infty}|{\delta}_{\rm av}(s)|$. 
Thus, using ${\delta}_{\rm av}\in\mathbf{H}_{\infty}$ and the aforementioned properties, by the theorem of 
real-rational approximation in $\mathbf{H}_{\infty}$~\cite[Theorem A.7.56]{CZ_1995},
we conclude that there exists {\ti a} ${\delta}\in\RHinf$ such that $\|{\delta} - {\delta}_{\rm av}\|_{\infty}$ 
is arbitrarily small.   
In the following, we will show that $\|{\delta}_{\rm av}\|_{\infty} \le r$ and ${\delta}_{\rm av} I$ 
stabilizes the cyclic network $\mathcal{F}_{\ell}(H,A)$.
Since $\|{\delta} - {\delta}_{\rm av}\|_{\infty}\approx 0$, we may therefore conclude that 
$\|{\delta}\|_{\infty}\le r+\epsilon = \|{\Delta}\|_{\infty}+\epsilon$
for an arbitrarily small $\epsilon>0$. Furthermore, note that the closed-loop system 
$[\![{\delta}I, \mathcal{F}_{\ell}(H,A)]\!]$ is equal to 
$[\![({\delta}-{\delta}_{\rm av})I, [\![ {\delta}_{\rm av}I, \mathcal{F}_{\ell}(H,A)]\!] \ ]\!]$. Therefore,
if $[\![ {\delta}_{\rm av}I, \mathcal{F}_{\ell}(H,A)]\!]$ is stable and $\|{\delta}-{\delta}_{\rm av}\|_{\infty}$ 
is arbitrarily small, we may conclude that ${\delta}I$ stabilizes $\mathcal{F}_{\ell}(H,A)$ by the small-gain theorem. 

To see $\|{\delta}_{\rm av}\|_{\infty} \le r$, note that given $s\in\mathbb{P}$, 
the imaginary part of $\frac{1}{n}\sum_{i=1}^n \log(1+{\delta}_i(s))$ belongs to $(-\pi,\pi]$ since those
of $\log(1+{\delta}_i(s))$, $i=1,\cdots,n$, do. Therefore by~\eqref{def:delta_av} we have 
$\log(1+{\delta}_{\rm av}(s))=\frac{1}{n}\sum_{i=1}^n \log(1+{\delta}_i(s))$. 
Given $s$ in CRHP, $|{\delta}_i(s)| \le r < 1$ for all $i=1,\cdots, n$; therefore all of 
$\log(1+{\delta}_i(s))$ belong to the convex set $\bfThe$ defined in~Lemma~\ref{lem:1}~(cf. equation~\eqref{def:bold_Theta}).  
The convexity of $\bfThe$ implies that $\log(1+{\delta}_{\rm av}(s))\in\bfThe$. 
Since $\log(1+\cdot)$ is a bijection between its domain and range, 
we must have $|{\delta}_{\rm av}(s)| \le r $, which in turn implies 
$\|{\delta}_{\rm av}\|_{\infty} \le r$.  

To see ${\delta}_{\rm av}I$ stabilizes the cyclic network $\mathcal{F}_{\ell}(H,A)$, note that by~\eqref{def:delta_av}, we have $(1+{\delta}_{\rm av}(s))^n = \prod_{i=1}^n(1+{\delta}_i(s))$ for any $s\in\mathbb{P}$. 
Therefore, for all $s\in\mathbb{P}$, 
\begin{eqnarray}
&& 1 - (-\mu)^nh(s)^n(1+{\delta}_{\rm av}(s))^n  \nonumber \\ 
&=& 1 - (-\mu)^nh(s)^n\prod_{i=1}^n(1+{\delta}_i(s)).
\label{char_polys}
\end{eqnarray}
Note that the left-hand side of~\eqref{char_polys} is the characteristic polynomial of 
$\Sigma({\delta}_{\rm av}I, H,A)$,
while the right-hand side is that of $\Sigma({\Delta}, H,A)$. 
Since ${\Delta}$ stabilizes the cyclic network, the roots of 
the right-hand side are in the ORHP, and therefore any $s\in\mathbb{P}$ 
such that $1+\mu^nh(s)^n(1+{\delta}_{\rm av}(s))^n=0$
must satisfy $\Re(s)<0$. On the other hand, if $s\not\in\mathbb{P}$, then $s$ would either be a pole for some ${\delta}_i$
or satisfy $1+{\delta}_i(s)=0$ for some $i$. In the second case, $1+{\delta}_{\rm av}(s)=0$ (by the definition of ${\delta}_{\rm av}$)
and thus $s$ is not a root of $1+\mu^nh(s)^n(1+{\delta}_{\rm av}(s))^n$. In the first case, by continuity, 
$|1+\mu^nh(\zeta)^n(1+{\delta}_{\rm av}(\zeta))^n|\to\infty$ as $\zeta\to s$, and therefore $s$ cannot be a root either. In conclusion, 
the roots of $1+\mu^nh(s)^n(1+{\delta}_{\rm av}(s))^n$ must be in the ORHP, and therefore ${\delta}_{\rm av}I$ stabilizes the cyclic
network.

\subsection{Proof of Lemma~\ref{lem:der_phi}}
\label{subsec:ProofLem_T1}
By the Gauss-Lucas theorem~\cite{Lucas:1879}, all zeros of $\dot{\phi}$ belong to the convex hull of the set of zeros of $\phi$.

\subsection{Proof of Lemma~\ref{lem:gcr_pcr_formulas}}
\label{subsec:ProofLem_T2}

Note that $\frac{d}{d\omega}\ln|\phi(j\omega)|$ and $\frac{d}{d\omega}\theta_{\phi}(\omega)$ are the real and imaginary parts of 
$\frac{d}{d\omega}\log\phi(j\omega)$, respectively, where $\log\phi(j\omega)$ is well-defined since $\phi$ has no zero on the imaginary {\ti axis}. 
Note that 
\begin{align*}
\frac{d}{d\omega}\log\phi(j\omega)&=\sum_{i=1}^n\frac{j}{j\omega-z_i} 
=\sum_{i=1}^n\frac{j}{2}\left(\frac{1}{j\omega-z_i}+\frac{1}{j\omega-\bar{z}_i}\right)\\
&=\sum_{i=1}^n\frac{-\omega+j(-\mathrm{Re}(z_i))}{(j\omega-z_i)(j\omega-\bar{z}_i)}.
\end{align*}
The second equality follows the fact that $z_i$ and $\bar{z}_i$ are both zeros of $\phi$, 
and is true regardless of $z_i$ being real or complex. The formulas in~\eqref{poly_loggcr} and~\eqref{poly_pcr} are 
straightforwardly derived from the last expression. 

\subsection{Proof of Lemma~\ref{lem:convexity_gen}}
\label{subsec:ProofLem_T3}

Let $\phi(j\omega)=x(\omega)+j y(\omega)$. Note that the assumption $\frac{d}{d\omega}\mathrm{Re}(\phi(j\omega))\not=0$
for $\omega\in(0,\omega_{\pi})$ implies that the inverse function $\chi = x\in\mathcal{I}_x \mapsto \omega\in(0,\omega_{\pi})$ exists, where $\mathcal{I}_x$ is the real interval $\{x(\omega)~|~\omega\in(0,\omega_{\pi})\}$. 
Therefore the contour $\{(x(\omega),y(\omega))~|~\omega\in(0,\omega_{\pi})\}$ is the graph of the function 
$h(x):=y(\chi(x))$. Moreover, the convexity of $\mathcal{R}_{\pi}(\phi)$ would be implied by the concavity of $h$. We will 
proceed to show that the conditions stated in the lemma {\ti imply} that $\frac{d^2}{dx^2}h(x) < 0$ for $x\in\mathcal{I}_x$.

First note that $\frac{d}{dx}h(x)=(\frac{dy}{d\omega})/(\frac{dx}{d\omega}) = \tan(\vartheta(\omega))$,
where $\vartheta(\omega):=\angle \frac{d}{d\omega}\phi(j\omega)=\angle (j\cdot \dot{\phi}(j\omega)) 
= \theta_{\dot{\phi}}(\omega)+\frac{\pi}{2}$. As such,
\begin{align*}
\frac{d^2}{dx^2}h(x) &= 
\frac{d}{d\omega}\tan\left(\theta_{\dot{\phi}}(\omega)+\frac{\pi}{2}\right) \left(\frac{dx}{d\omega}\right)^{-1}\\
&= \left(\frac{dx}{d\omega}\right)^{-1} \sec\left( \theta_{\dot{\phi}}(\omega)+\frac{\pi}{2} \right)^2
\left(\frac{d}{d\omega}\theta_{\dot{\phi}}(\omega)\right){\ti.}
\end{align*}
The above formula shows that $\frac{d^2}{dx^2}h(x) < 0$ for $x\in\mathcal{I}_x$ if 
$\frac{d}{d\omega}\theta_{\dot{\phi}}(\omega) > 0$ and $\frac{d}{d\omega}\mathrm{Re}(\phi(j\omega)) < 0$ for $\omega\in(0,\omega_{\pi})$.

\subsection{Proof of Lemma~\ref{lem:tech_cond4}}
\label{subsec:ProofLem_T4}

Recall that inequality~\eqref{ineq:gen_pcr_generic} has the equivalent form stated in~\eqref{ineq:eqiv_cond3} in the polar coordinate, which in turn
is equivalent to
\begin{align}
r_{\phi_i}'/r_{\phi_i} \le \omega\left( \left(r_{\phi_i}'/r_{\phi_i}\right)^2 +(\theta_{\phi_i}')^2\right),
\label{ineq:eqiv_cond4}
\end{align}
where $r_{\phi_i}$ and $\theta_{\phi_i}$ are respectively the gain and phase of $\phi_i$, $i=1,2$. 

Now let $\phi:=r_{\phi}\mathrm{e}^{j\theta_{\phi}}$, where $r_{\phi}= r_{\phi_1}r_{\phi_2}$ and $\theta_{\phi}=\theta_{\phi_1}+\theta_{\phi_2}$. We have
\begin{align*}
\omega&\left(\left(r_{\phi}'/r_{\phi}\right)^2+(\theta_{\phi}')^2\right)-r_{\phi}'/r_{\phi}=\\
&\sum_{i=1, 2}\omega\left( \left(r_{\phi_i}'/r_{\phi_i}\right)^2+(\theta_{\phi_i}')^2\right)-r_{\phi_i}'/r_{\phi_i}\\
&\hspace{0.7cm} +2\omega\left(\left(r_{\phi_1}' r_{\phi_2}' \right)/\left( r_{\phi_1} r_{\phi_2} \right) +\theta_{\phi_1}'\theta_{\phi_2}'\right){\ti.}
\end{align*}
Since the three terms on the right-hand side of the equality above are all non-negative, we conclude that $\phi$ satisfies~\eqref{ineq:eqiv_cond4}, 
and hence~\eqref{ineq:gen_pcr_generic}.
This concludes the proof.

\subsection{Proof of Lemma~\ref{lem:tech_cond5}}
\label{subsec:ProofLem_T5}

The result follows from a straightforward calculation: let $F(\Omega):=|\phi(j\omega)|^2 = |\phi_1(j\omega)\phi_2(j\omega)|^2 = F_1(\Omega)F_2(\Omega)$. 
Then $F'' = F_1''F_2+F_1F_2'' + 2F_1'F_2'$. Since $\phi_1$ and $\phi_2$ are gain-monotone increasing and satisfy~\eqref{ineq:2nd_deriv}, 
clearly all terms on the right-hand side of the aforementioned equation are non-negative. Thus, $F''\ge 0$ as claimed.

\begin{IEEEbiography}%
{Shinji Hara}
received the Ph.D. from Tokyo Inst. of Tech. (TITech), in 1981. 
In 1984, he joined TITech 
and had served as a Professor for ten years. 
From 2002 to 2017 he was a Professor in the Dept. of 
Information Physics and Computing at the University of Tokyo (U Tokyo). 
He is Professor Emeritus of TITech and U Tokyo. 
His current research interests include 
networked dynamical systems and glocal control.  
Dr. Hara received many awards including 
the George S. Axelby Outstanding Paper Award from the IEEE CSS in 2006. 
He was the President of SICE, Japan in 2009, a Vice President of the IEEE CSS in 2009-2010, and 
an IFAC Council member in 2011-2017. He is a Fellow of IFAC, IEEE, ACA and SICE. 
\end{IEEEbiography}

\vspace{-10mm}

\begin{IEEEbiography}%
{Yutaka Hori}
received the B.S degree in engineering, and the M.S. and Ph.D. degrees in information science and technology from the University of Tokyo 
in 2008, 2010 and 2013, respectively. He held a postdoctoral appointment at California Institute of Technology from 2013 to 2016. In 2016, 
he joined Keio University, where he is currently an associate professor. His research interests lie in feedback control theory and its 
applications to synthetic biomolecular systems. He is a member of IEEE, SICE, and ISCIE. 
\end{IEEEbiography}

\vspace{-10mm}

\begin{IEEEbiography}%
{Tetsuya Iwasaki} (M'90-SM'01-F'09) received B.S.\ and M.S.\
degrees %
from the Tokyo Institute of Technology in 1987 and 1990, %
and his Ph.D.\ degree %
from Purdue University in 1993. He held faculty positions at Tokyo Tech and
University of Virginia before joining the UCLA.
His research interests include {\ti control, oscillation, locomotion, pattern formation,
and neurodynamics.} 
He has received several awards from NSF, SICE, IEEE, and ASME.
He has served as Senior/Associate Editor of several control journals.
\end{IEEEbiography}

\vspace{-10mm}
\begin{IEEEbiography}%
{Chung-Yao Kao} 
received the Sc.D. degree from the Massachusetts Institute of Technology, Cambridge, USA, in 2002.
From 2002 to 2004, he held postdoctoral research positions at the Lund Institute of Technology, Lund, Sweden, the 
Mittag-Leffler Institute, Djursholm, Sweden, and the Royal Institute of Technology (KTH), Stockholm, Sweden. 
From 2004 to 2009, he was a Senior Lecturer with 
the Department of Electrical and Electronic Engineering, University of Melbourne, Australia. In 2009, he 
joined the Department of Electrical Engineering, National Sun Yat-Sen University, Kaohsiung, Taiwan, where he is 
currently a Professor. His research interests include robust control, systems theory, and optimization. 
\end{IEEEbiography}

\vspace{-10mm}

\begin{IEEEbiography}%
{Sei Zhen Khong}
received the Bachelor of Electrical Engineering degree (with first class honours) and the Ph.D. degree from The 
University of Melbourne, Australia,
in 2008 and 2012, respectively. He has held research positions at the Department of Electrical and Electronic Engineering, 
The University of
Melbourne, Australia, the Department of Automatic Control, Lund University, Sweden, the Institute for Mathematics and 
its Applications, The University
of Minnesota, Twin Cities, USA, and the Department of Electrical and Electronic Engineering, The University of Hong Kong, 
China. His research interests
include network control, robust control, systems theory, and extremum seeking control.
\end{IEEEbiography}


\begin{thebibliography}{1}
%
\bibitem{Elowitz2000}
M.~B. Elowitz and S. Leibler, ``A synthetic oscillatory network of transcriptional regulators,''
\emph{Nature}, vol. 403, no. 6767, pp. 335--338, 2000.

\bibitem{FHNmodel}
R. FitzHugh, ``Impulses and physiological states in theoretical models of nerve membrane,''
{\em Biophysical Journal}, Vol.1, pp.445-466, 1961.

\bibitem{ijspeert:08} %
A.J. Ijspeert,
\newblock ``Central pattern generators for locomotion control in animals and robots: A review,''
\newblock {\em Neural Networks}, vol.~21, pp.642-653, 2008.

\bibitem{YMKH:MBMC2015}
Y.~Hori, H.~Miyazako, S.~Kumagai, S.~Hara, 
\newblock ``Coordinated spatial pattern formation in biomolecular communication networks,''
\newblock {\em IEEE Trans. on Molecular, Biological and Multi-Scale Communications}, 
vol.~1 no. 2, pp. 111--121, 2015.

\bibitem{hara2022instability}
S.~Hara, T.~Iwasaki, and Y.~Hori, ``Instability margin analysis for
  parametrized {LTI} systems with application to repressilator,'' {\em
  Automatica}, vol.~136, no. 2, 110047, 2022.

\bibitem{KHHIK:IFAC23}
C.-Y. Kao, S. Hara, Y. Hori, T. Iwasaki, and S.Z. Khong,
``On phase change rate maximization with practical applications,''
{\em IFAC World Congress}, vol.~56, no. 2, pp.5805-5810, 2023.

\bibitem{hara2023exact}
S.~Hara, C.-Y. Kao, S.~Khong, T.~Iwasaki, and Y.~Hori, ``Exact instability margin analysis and minimum-norm strong 
stabilization -- phase change rate maximization,'' {\em IEEE Transactions on Automatic Control},
vol.~69, no.~4, pp.~2084--2099, 2024.
%

\bibitem{Youla:Automatica1974}
D.C. Youla, J.J. Bongiorno, Jr., and C.N. Lu, 
``Single-loop feedback-stabilization of linear multivariable dynamical plants,''
{\em Automatica}, vol.10, pp.159-173, 1974.

\bibitem{hara:19}
S.~Hara, T.~Iwasaki, and Y.~Hori.
\newblock Robust stability analysis for {LTI} systems with generalized
  frequency variables and its application to gene regulatory networks.
\newblock {\em Automatica}, 105:96--106, 2019.

{\red 
\bibitem{NetworkRIR_short} 
S.~Hara, et al.  ``Robust Instability Analysis for Networked Dynamical Systems: Upper and Lower Bounds,'' 
{\em IEEE Transactions on Automatic Control}, (to be submitted)
}

\bibitem{Niederholtmeyer2015}
H. Niederholtmeyer, Z. Z. Sun, Y. Hori, E. Yeung, A. Verpoorte, R. M. Murray and S. J. Maerkl, ``Rapid Cell-free Forward Engineering of Novel Genetic Ring Oscillators,'' {\em eLife}, vol.4, e09771, 2015.

\bibitem{Vinnicombe2001}
G.~ Vinnicombe, 
``Uncertainty and Feedback: $H_{\infty}$ loop-shaping and the $\nu$-gap metric,'' 
{\em Imperial College Press}, 2001.

%
%
%
%

\bibitem{Lucas:1879}
F.~Lucas, ``Sur une application de la m\'{e}canique rationnelle \`{a} la   th\'{e}orie des \'{e}quations,'' 
\emph{C. R. Acad. Sci. Paris} 89, 224–226, 1879.

%
%
%
%
%
%

\bibitem{CZ_1995}
R.~F.~Curtain, and H.~Zwart, 
``An Introduction to Infinite-Dimensional Linear Systems Theory,''
\emph{Springer}, 1995.

\bibitem{Inoue:ECC2013}
M. Inoue, et al., 
``An instability condition for uncertain systems toward robust bifurcation analysis,'' 
{\em Proc. of the European Control Conf. 2013}, pp. 3264-3269, 2013.

\bibitem{Hori2011}
Y. Hori, T.-H. Kim, S. Hara, ``Existence criteria of periodic oscillations in cyclic gene regulatory networks,'',
{\em Automatica},vol. 47, no. 6, pp. 1203--1209, 2011.

\end{thebibliography}
\end{document}